\documentclass{article}
\usepackage[margin=2.5cm]{geometry}

\usepackage{pdfcolfoot}
\usepackage{lipsum}
\usepackage{amsmath}
\usepackage{amssymb}
\usepackage{amsfonts}
\usepackage{bbm}
\usepackage{algorithm}
\usepackage{algpseudocode}
\usepackage{mathtools}
\usepackage{soul}
\usepackage[normalem]{ulem}
\usepackage{cancel}

\usepackage{amsthm} 

\usepackage{ifthen}
\usepackage[stamp=false]{draftwatermark}
\SetWatermarkText{\textsf{DRAFT}}
\SetWatermarkScale{7}
\SetWatermarkColor[gray]{0.93}
\usepackage{xstring}

\floatstyle{ruled}
\newfloat{algo}{htbp}{algo}
\floatname{algo}{Algorithm}
\usepackage{hyperref}
\usepackage{cleveref}

\usepackage{numbertabbing}

\usepackage{IEEEtrantools}
\allowdisplaybreaks

\usepackage[shortlabels,inline]{enumitem}

\usepackage{tikz}
\usetikzlibrary{calc}
\usetikzlibrary{arrows}
\usetikzlibrary{arrows.meta}
\usetikzlibrary{patterns}
\usetikzlibrary{positioning}
\usetikzlibrary{decorations.pathreplacing}
\usetikzlibrary{shapes.misc}
\usetikzlibrary{spy}

\usepackage{pgfplots}
\pgfplotsset{compat=1.14}

\usepackage{xspace}

\newcommand{\ie}[0]{\emph{i.e.}\xspace}

\crefname{algo}{algorithm}{algorithms}

\newcounter{countersetupcrefforlines}
\loop
  \edef\temp{\noexpand\crefname{linecounter\arabic{countersetupcrefforlines}}{line}{lines}}
  \temp
  \stepcounter{countersetupcrefforlines}
  \ifnum\value{countersetupcrefforlines}<11
\repeat

\makeatletter
\newcommand{\customlabel}[3][]{%
   \protected@write \@auxout {}{\string \newlabel {#2}{{#3}{\thepage}{#3}{#2}{}} }%
   \protected@write \@auxout {}{\string \newlabel {#2@cref}{{[#1][#3][]#3}{[1][\thepage][]\thepage}}}%
   \hypertarget{#2}{}%
}
\makeatother

\usepackage{pifont}

\newcommand{\txpool}[0]{\ensuremath{\mathit{txpool}}}
\newcommand{\tx}[0]{\ensuremath{\mathit{tx}}}
\newcommand{\chain}[0]{\ensuremath{\mathsf{ch}}}
\newcommand{\chainfrozen}[1][]{\ensuremath{\mathsf{ch}^{\mathrm{frozen}\ifthenelse{\equal{#1}{}}{}{,#1}}}}
\newcommand{\Chain}[0]{\ensuremath{\mathsf{Ch}}}
\newcommand{\chainava}[0]{\ensuremath{\mathsf{chAva}}}
\newcommand{\chainfin}[0]{\ensuremath{\mathsf{chFin}}}

\newcommand{\chaincanmfc}[0]{{\chain^{\mathsf{MFC}}}}

\newcommand{\fastcands}[0]{\ensuremath{\mathsf{fast}^\mathrm{cands}}}
\newcommand{\fastcand}[0]{\ensuremath{\mathsf{fast}^\mathrm{cand}}}

\newcommand{\mfc}[0]{\ensuremath{\operatorname{\textsc{MFC}}}}
\newcommand{\mfcvote}[1]{\ensuremath{\mfc^{\voting{#1}}}}
\newcommand{\mfcpropose}[1]{\ensuremath{\mfc^{\proposing{#1}}}}
\newcommand{\specifyExec}[2]{\overset{#1}{#2}{}}
\newcommand{\mfcFFG}{\specifyExec{\FFGExec}{\mfc}{}}
\newcommand{\mfcNoFFG}{\specifyExec{\NoFFGExec}{\mfc}{}}
\newcommand{\mfcvoteFFG}[1]{\ensuremath{\mfcFFG^{\voting{#1}}}}
\newcommand{\mfcproposeFFG}[1]{\ensuremath{\mfcFFG^{\proposing{#1}}}}
\newcommand{\mfcvoteNoFFG}[1]{\ensuremath{\mfcNoFFG^{\voting{#1}}}}
\newcommand{\mfcproposeNoFFG}[1]{\ensuremath{\mfcNoFFG{}^{\proposing{#1}}}}

\newtheorem{theorem}{Theorem}

\newtheorem{lemma}{Lemma}

\newtheorem{definition}{Definition}

\newcommand{\mfcproposeExec}[2]{\ensuremath{\specifyExec{#1}{\mfc}^{\proposing{#2}}}}

\crefname{condition}{condition}{conditions}

\newcommand{\VFFG}{\overset{\FFGExec}{\V}{}}
\newcommand{\VNoFFG}{\overset{\NoFFGExec}{\V}{}}
\newcommand{\makeNoFFG}[2][]{\specifyExec{\NoFFGExec}{\ifthenelse{\equal{#1}{}}{#2}{#2_{#1}}}{}}
\newcommand{\makeFFG}[2][]{\specifyExec{\FFGExec}{\ifthenelse{\equal{#1}{}}{#2}{#2_{#1}}}{}}
\ExplSyntaxOn
\cs_generate_variant:Nn \tl_if_eq:nnTF { ee }
\NewDocumentCommand{\removeparens}{m}
{
  \tl_if_eq:eeTF {\tl_head:n {#1}} { ( } {\tl_if_eq:eeTF {\tl_range:nnn {#1} {-1} {-1}} { ) } {\tl_range:nnn {#1} {2} {-2}} {#1}} {#1}
}

\ExplSyntaxOff
\newcommand{\proposing}[1]{{\ensuremath{\mathsf{propose}(\removeparens{#1})}}}
\newcommand{\voting}[1]{{\ensuremath{\mathsf{vote}(\removeparens{#1})}}}
\newcommand{\fastconfirming}[1]{{\ensuremath{\mathsf{fconf}(\removeparens{#1})}}}
\newcommand{\merging}[1]{{\ensuremath{\mathsf{merge}(\removeparens{#1})}}}

\newcommand{\C}[0]{\ensuremath{\mathcal{C}}}
\newcommand{\T}[0]{\ensuremath{\mathcal{T}}}
\newcommand{\J}[0]{\ensuremath{\mathcal{J}}}
\newcommand{\calS}[0]{\ensuremath{\mathcal{S}}}
\newcommand{\GJ}[0]{\ensuremath{\mathcal{GJ}}}
\newcommand{\GJfrozen}[1][]{\ensuremath{\GJ^{\mathrm{frozen}\ifthenelse{\equal{#1}{}}{}{,#1}}}}
\newcommand{\GF}[0]{\ensuremath{\mathcal{GF}}}

\newcommand{\slot}[0]{\ensuremath{\operatorname{\mathrm{slot}}}}
\newcommand{\wop}[0]{w.o.p.\xspace}

\long\def\blockcomment#1\endblockcomment{}
\algdef{SE}[DOWHILE]{Do}{doWhile}{\algorithmicdo}[1]{\algorithmicwhile\ #1}%

\algdef{SE}[UPON]{Upon}{EndUpon}[1]{\textbf{upon}\ $#1$}{\textbf{end upon}}%
\algdef{SE}[AT]{At}{EndAt}[1]{\textbf{at}\ $#1$}{\textbf{end at}}%

\newcommand{\GAT}[0]{\ensuremath{\mathsf{GAT}}}
\newcommand{\GST}[0]{\ensuremath{\mathsf{GST}}}

\newcommand{\V}[0]{\ensuremath{\mathcal{V}}}
\newcommand{\Vglobal}[0]{\ensuremath{\V_\mathsf{G}}}
\newcommand{\Vfrozen}[1][]{\ensuremath{\mathcal{V}^{\mathrm{frozen}\ifthenelse{\equal{#1}{}}{}{,#1}}}}
\newcommand{\X}[0]{\ensuremath{\mathcal{X}}}
\newcommand{\M}[0]{\ensuremath{\mathcal{M}}}

\newcommand{\FC}[0]{\ensuremath{\mathsf{FC}}}

\newcommand{\negl}[0]{\ensuremath{\operatorname{negl}}}

\newcommand{\Tconf}[0]{\ensuremath{T_\mathsf{conf}}}
\newcommand{\Tafter}[0]{\ensuremath{T_\mathsf{sec}}}
\newcommand{\Tdyn}[0]{\ensuremath{T_\mathsf{dyn}}}
\newcommand{\treorg}[0]{\ensuremath{{t_\mathsf{reorg}}}}
\newcommand{\Treorg}[0]{\ensuremath{{T_\mathsf{reorg}}}}

\newcommand{\genesis}[0]{\ensuremath{B_\mathrm{genesis}}\xspace}

\newcommand{\LMDGHOST}[0]{\textsf{LMD-GHOST}\xspace}

\newcommand{\TOBSVD}[0]{\textsf{TOB-SVD}\xspace}

\newcommand{\protocol}[0]{\textsf{Majorum}\xspace}

\newcommand{\FIL}[0]{\textsf{FIL}\xspace}

\newcommand{\LOGbft}[2]{%
    \ifthenelse{\equal{#1}{}}{%
        \ensuremath{\mathsf{LOG}_{\mathrm{bft}}^{#2}}%
    }{%
        \ensuremath{\mathsf{LOG}_{\mathrm{bft},#1}^{#2}}%
    }%
}

\newcommand{\ld}[1]{%
    \ifthenelse{\equal{#1}{}}{%
        \ensuremath{\mathrm{L}^{(c)}}%
    }{%
        \ensuremath{\mathrm{L}^{(#1)}}%
    }%
}

\newcommand{\bprop}[1]{%
    \ifthenelse{\equal{#1}{}}{%
        \ensuremath{\Hat{b}}%
    }{%
        \ensuremath{\Hat{b}_{#1}}%
    }%
}



\begin{document}

\title{\protocol: Ebb-and-Flow Consensus with Dynamic Quorums}

\author{Francesco D'Amato\\
  Ethereum Foundation\\
  \url{francesco.damato@ethereum.org}
  \and
  Roberto Saltini\\
  Ethereum Foundation\\
  \url{roberto.saltini@ethereum.org}
  \and
  Thanh-Hai Tran\\
  Independent Researcher\\
  \url{thanhhai1302@gmail.com}
  \and
  Yann Vonlanthen\\
  Ethereum Foundation\\
  \url{yann.vonlanthen@ethereum.org}
  \and
  Luca Zanolini\\
  Ethereum Foundation\\
  \url{luca.zanolini@ethereum.org}
}
\date{}

\maketitle
\thispagestyle{plain}
\pagestyle{plain}

\begin{abstract}
Dynamic availability is the ability of a consensus protocol to remain live despite honest participants going offline and later rejoining. A well-known limitation is that dynamically available protocols, on their own, cannot provide strong safety guarantees during network partitions or extended asynchrony. Ebb-and-flow protocols [SP21] address this by combining a dynamically available protocol with a partially synchronous finality protocol that irrevocably finalizes a prefix.

We present \protocol, an ebb-and-flow construction whose dynamically available component builds on a quorum-based protocol (\TOBSVD{}). Under optimistic conditions, \protocol finalizes blocks in as few as three slots while requiring only a \emph{single} voting phase per slot. In particular, when conditions remain favourable, each slot finalizes the next block extending the previously finalized one.
\end{abstract}

\section{Introduction and Related Work}
\label{sec:introduction}

Traditional Byzantine consensus protocols such as PBFT~\cite{DBLP:conf/osdi/CastroL99} and HotStuff~\cite{DBLP:conf/podc/YinMRGA19} typically assume a fixed set of validators that remain online; in such settings, going offline is treated as a fault.

Some permissionless systems require \emph{dynamic participation}: honest validators may go offline and later rejoin. This notion was implicit in Bitcoin~\cite{nakamoto2008bitcoin} and formalized by Pass and Shi via the \emph{sleepy model}~\cite{sleepy}. We call a protocol \emph{dynamically available} if it preserves safety and liveness under this type of churn.

A limitation is that dynamically available protocols, on their own, cannot provide strong safety guarantees during network partitions or extended asynchrony while also maintaining liveness~\cite{DBLP:journals/corr/abs-2304-14701}. Ebb-and-flow protocols~\cite{DBLP:conf/sp/NeuTT21} address this by combining a dynamically available chain-growth protocol with a partially synchronous finality layer that irrevocably finalizes a prefix. Ethereum follows this pattern: \LMDGHOST~\cite{zamfir} provides dynamically available chain growth, and Casper FFG~\cite{casper} finalizes a prefix (as in Gasper~\cite{gasper}).

Momose and Ren~\cite{DBLP:conf/ccs/Momose022} initiated the study of \emph{deterministically safe} dynamically available consensus via quorum techniques adapted to dynamic participation. This line of work led to multiple quorum-based dynamically available protocols with different trade-offs~\cite{DBLP:journals/iacr/MalkhiMR22,DBLP:conf/ccs/MalkhiM023,DBLP:conf/wdag/GafniL23,DBLP:conf/podc/DAmatoLZ24}.

In particular, D'Amato \emph{et al.}~\cite{streamliningSBFT} introduced \TOBSVD{}, a majority-quorum-based protocol targeting practicality at scale. \TOBSVD{} tolerates up to $50\%$ Byzantine faults and can decide with a single voting phase in the best case, improving over earlier multi-phase designs.

In this work we present \protocol, an ebb-and-flow protocol whose dynamically available component builds on \TOBSVD{}. Under optimistic conditions, \protocol achieves three-slot finality while requiring only a single voting round per slot. Relative to the only ebb-and-flow construction we are aware of that is deployed in production -- namely, Ethereum’s -- this reduces optimistic time to finality from 64 slots to 3. We further introduce, to our knowledge for the first time, a principled integration of two complementary confirmation rules within an ebb-and-flow protocol: traditional $\kappa$-deep confirmation and \emph{fast confirmation}.

\subsection{Technical Overview}

In large-scale consensus systems, each voting round is costly. With hundreds of thousands of validators, votes are typically aggregated and re-broadcast before they can influence the global view; consequently, even a single \emph{logical} voting phase can incur an end-to-end delay closer to $2\Delta$ than $\Delta$. Ethereum illustrates this effect: attestations are aggregated and disseminated before contributing to fork choice, yielding an effective $2\Delta$ latency per voting round. This motivates our first design goal: a dynamically available protocol that, under synchrony, makes progress with as few voting phases as possible.

Our second goal is to obtain the guarantees characteristic of ebb-and-flow protocols~\cite{ebbandflow}. Under synchrony (even with reduced participation), the system should remain live and produce a synchronously safe chain. Under a network partition or extended asynchrony, a \emph{finalized} prefix should remain forever safe, while only the non-finalized suffix may be reverted. Achieving both goals calls for a dynamically available protocol with an efficient voting structure that can be paired cleanly with a partially synchronous finality gadget.

For these reasons, we take \TOBSVD{} as our starting point. \TOBSVD{} is a dynamically available consensus protocol based on majority quorums; in the synchronous model, it achieves optimal Byzantine resilience within this design space and, in the best case, requires only a single voting phase per decision. In the remainder of this overview, we describe the modifications we apply to \TOBSVD{} and how we pair the resulting dynamically available protocol with a partially synchronous finality layer whose purpose is to \emph{finalize} -- that is, to provide asynchronous safety to -- a prefix of the produced chain, yielding \protocol.

The appeal of \protocol, and of ebb-and-flow protocols more generally, is that they can remain live under synchrony even in low participation regimes -- for instance, when a large fraction of validators are temporarily offline. Ethereum is a noteworthy deployed example of this design pattern: block production can continue under synchrony even when participation drops (although finality may stall). Moreover, Ethereum’s incentive structure encourages validators to stay online and participate promptly; nevertheless, dynamic participation remains practically relevant, for example due to crash failures caused by consensus-client bugs~\cite{Nijkerk2023LossOfFinality}.

\TOBSVD{} proceeds in \emph{slots}: at the start of slot~$t$, an elected proposer broadcasts a block (propose phase); after~$\Delta$, validators cast votes (vote phase); and after another $\Delta$, the protocol decides according to the majority of observed votes (decide phase). In the original protocol, proposals made in slot~$t$ are decided during the decide phase of slot~$t+1$. All safety guarantees are fully deterministic; no probabilistic arguments are used. Moreover, in \TOBSVD{}, proposals made by honest proposers will always be voted on by all awake honest validators.

A key insight is that the deterministic decision phase does not influence validator behavior; rather, it acts as an explicit confirmation mechanism. 
We therefore remove the decision phase and confirm chains directly during the vote phase via a \emph{$\kappa$-deep confirmation rule}, where $\kappa$ is chosen so that, within $\kappa$
slots, an honest proposer is elected with overwhelming probability, obtaining a variant that provides \emph{probabilistic} rather than deterministic safety (Section~\ref{sec:revisiting-tob}). However, this modification alone confirms only those chains whose proposals lie at least $\kappa$ slots in the past. To address this limitation, we introduce an optimistic mechanism for \emph{fast confirmations} (Section~\ref{sec:prob-tob}). Under favorable conditions, a chain can be confirmed in the very slot in which it is proposed: concretely, if all validators are awake and more than $\tfrac{2}{3}n$ of them vote for the same chain, that chain is \emph{fast-confirmed} and may be immediately adopted in place of the $\kappa$-deep prefix.

This new protocol constitutes the dynamically available component of the ebb-and-flow construction \protocol: once a block is confirmed in the dynamically available layer, this confirmation can be fed into the partially synchronous layer as a candidate for finalization.

Integrating these two layers is straightforward. We extend the voting message format of \TOBSVD{} by adding an additional component that closely mirrors the mechanism of Casper FFG -- the finality gadget used in Ethereum’s consensus protocol. Concretely, the voting phase of slot~$t$ serves two purposes: (1) voting for the proposal made in slot~$t$, and (2) contributing to the finalization of a previously confirmed chain within \TOBSVD{}.

Under this design, a proposal made by an honest proposer in slot~$t$ is optimistically fast-confirmed during the same slot. It then becomes the \emph{target} of the first (of two) voting rounds in slot~$t+1$. If it gathers more than $\tfrac{2}{3}n$ votes in this first round, we say that it is \emph{justified}. Finally, it becomes finalized in slot~$t+2$, during the second round of voting. Figure~\ref{fig:overview} illustrates the progression of the combined protocol. Although this represents roughly a 3$\times$ increase over same-slot finality, it avoids the need for multiple voting rounds within a single slot -- an approach that, at scale, would necessitate longer slot durations -- while retaining a single voting round per slot and thereby simplifying the protocol’s timing assumptions and overall design.

The remainder of this paper is organized as follows. Section~\ref{sec:model} presents the model and background. Section~\ref{sec:da-component} introduces the dynamically available component, reviewing \TOBSVD{} (Section~\ref{sec:recalling-tob}) and extending it to probabilistic safety and fast confirmations (Section~\ref{sec:revisiting-tob}; proofs in Appendix~\ref{appendix:analysis-prob-ga-fast}). Section~\ref{sec:ffg} presents the finality gadget (proofs in Appendix~\ref{appendix:proofs}). Section~\ref{sec:ga-based} gives the full ebb-and-flow protocol and its security analysis (Appendix~\ref{appendix:analysis-3sf}). We conclude in Section~\ref{sec:conclusion}.

\section{Model and Preliminary Notions}
\label{sec:model}

\subsection{System Model}

\noindent\textbf{Validators.}  We consider a set of $n$ \emph{validators} $v_1, \dots, v_n$, each with a unique cryptographic identity and known public key. Validators execute a common protocol. An adaptive probabilistic poly-time adversary $\mathcal{A}$ may corrupt up to $f$ validators, gaining access to their state. Once corrupted, a validator is \emph{adversarial}; otherwise it is \emph{honest}. Honest validators follow the protocol, while adversarial ones behave arbitrarily. All validators have equal \emph{stake}, and a proven misbehaving validator is \emph{slashed} (loses part of its stake). 

\noindent\textbf{Network.}  We assume partial synchrony: validators have synchronized clocks and message delays are initially unbounded, but after an unknown \emph{global stabilization time} (\GST) delays are bounded by $\Delta$. 
  
\noindent\textbf{Time.}  Time is divided into discrete \emph{rounds}. A \emph{slot} consists of $4\Delta$ rounds. For round~$r$, its slot is $\slot(r) := \lfloor r / 4\Delta \rfloor$.  

\noindent\textbf{Proposer Election Mechanism.}  In each slot $t$, a validator $v_p$ is chosen as \emph{proposer}, denoted $v_p^t$. If $t$ is clear, we omit the superscript. The proposer election mechanism must satisfy: \emph{Uniqueness} (exactly one proposer per slot), \emph{Unpredictability} (identity hidden until revealed), and \emph{Fairness} (each validator elected with probability $1/n$).

\noindent\textbf{Transactions Pool.}  We assume the existence of an ever growing set of transaction, called \emph{transaction pool} and denoted by $\txpool$, that any every validator has read access to. Consistently with previous notation, we use $\txpool^r$ to refer to the content of $\txpool$ at time~$r$.  

\noindent\textbf{Blocks and Chains.}  A \emph{block} is $B = (b, p)$, where $b$ is the block body (a batch of transactions) and $p \geq 0$ is the slot when $B$ is proposed. The \emph{genesis block} is $\genesis = (b_{-1}, -1)$ and has no parent. Each block references its parent, and a child must have a larger slot. A block induces a \emph{chain}, the path to genesis. Chains are identified with their tip blocks. We write $\chain_1 \prec \chain_2$ if $\chain_1$ is a strict prefix of $\chain_2$, and say that two chains \emph{conflict} if neither extends the other.  

The $\kappa$-deep prefix of chain $\chain$ at slot $t$, denoted $\chain^{\lceil \kappa, t}$, is the longest prefix $\chain'$ with $\chain'.p \leq t - \kappa$. When $t$ is clear, we write $\chain^{\lceil \kappa}$. We define $\chain_1 \leq \chain_2$ if $\chain_1.p \leq \chain_2.p$, with equality only if the chains are identical. We call $\chain.p$ the \emph{height} (or \emph{length}) of $\chain$. A transaction $\tx$ belongs to $\chain$ if it appears in a block of $\chain$, written $\tx \in \chain$. Similarly, $\mathit{TX} \subseteq \chain$ means all transactions in $\mathit{TX}$ are included in $\chain$.  

\noindent\textbf{Proposing and Voting.} Validators may \emph{propose} chains, generalizing block proposals. Proposals occur in the first round of each slot. For readability, define $\proposing{t} := 4\Delta t$. Each slot also has a voting round, where validators \emph{vote} on proposed chains. Define $\voting{t} := 4\Delta t + \Delta$.

\noindent\textbf{Gossip.} We assume a best-effort gossip primitive that eventually delivers messages to all validators. Honest-to-honest messages are always delivered and cannot be forged. Once an adversarial message is received by an honest validator $v_i$, it is also gossiped by $v_i$. Votes and blocks are always gossiped, whether received directly or inside another message. A validator receiving a vote also gossips the corresponding block, and receiving a proposal gossips its blocks and votes. A proposal for slot $t$ is gossiped only during the first $\Delta$ rounds of slot~$t$.  
 
\noindent\textbf{Views.} A \emph{view} at round $r$, denoted $\V^r$, is the set of all messages a validator has received by round $r$. We write $\V_i^r$ for the view of validator $v_i$.  
 
\noindent\textbf{Sleepiness.} In each round $r$, the adversary may declare which honest validators are \emph{awake} or \emph{asleep}. Asleep validators pause protocol execution; messages are queued until they wake. Upon waking at round $r$, a validator must complete a \emph{joining protocol} before becoming \emph{active}~\cite{goldfish} (see the paragraph below). Adversarial validators are always awake. Thus, \emph{awake}, \emph{asleep}, and \emph{active} apply only to honest validators. The sleep schedule is adversarially adaptive. In the sleepy model~\cite{sleepy}, awake and active coincide. Eventually, at an unknown \emph{global awake time} (\GAT), all validators stay awake forever.   

Let $H_r$, $A_r$, and $H_{r,r'}$ denote, respectively, the set of active validators at round $r$, the set of adversarial validators at round $r$, and the set of validators that are active \emph{at some point} in rounds $[r,r']$, i.e., 
$H_{r,r'} = \bigcup_{i=r}^{r'} H_i, \text{with } H_i \coloneqq \emptyset \text{ if } i < 0.$
We let $A_\infty = \lim_{t \to \infty} A_{\voting{t}}$ and note that $f = |A_\infty|$.  
Unless otherwise stated, we assume throughout that $f < \frac{n}{3}$.  For every slot $t$ after \GST, we require

\begin{align}
  &|H_{\voting{(t-1)}} \setminus A_{\voting{t}}| > |A_{\voting{t}}|
  \label{eq:pi-sleepiness}
\end{align}

In words, the number of validators active at round $\voting{(t-1)}$ and not corrupted by $\voting{t}$ must exceed the sum of adversarial validators at $\voting{t}$.

We impose this restriction to limit the number of honest validators the adversary can corrupt after the voting phase of slot~$t-1$ but before those votes are tallied in later slots. Without such a bound, the adversary could corrupt \emph{too many} honest validators immediately after the voting phase and force them to equivocate. 

The main difference from the original sleepy model is that our constraints apply only to \emph{active} validators, i.e., those that have completed a joining protocol and therefore have been awake for at least $\Delta$ rounds before \(\voting{t}\). This reflects the \(2\Delta\) stabilization period assumed by \TOBSVD{} for security; we refer to the original paper for a formal treatment~\cite{streamliningSBFT}. In practice, a validator is unlikely to enjoy buffered and instantaneous delivery of messages as soon as it wakes up: it must request recovery and wait for peers to return any decision-relevant messages it missed. We treat a validator as active only after these replies arrive, which takes at least \(2\Delta\) (one \(\Delta\) to send, one \(\Delta\) to receive) and may be much longer when substantial state must be fetched. Hence, an explicit \(2\Delta\) stabilization window is typically negligible relative to practical recovery times.


We call the adversary model defined above the \emph{generalized partially synchronous sleepy model}.  
When the context is clear, we refer to it simply as the \emph{sleepy model}.  
An execution in this model is said to be \emph{compliant} if and only if it satisfies the sleepiness Condition~\eqref{eq:pi-sleepiness}.

\noindent\textbf{Joining Protocol.}  
An honest validator waking at round $r$ must first run the \emph{joining protocol}. If $v_i$ wakes in round $r \in (\voting{(t-2)}+\Delta,\, \voting{(t-1)}+\Delta]$, it receives all queued messages and resumes execution, but sends nothing until round $\voting{t}$, when it becomes active (unless corrupted or put back to sleep). If $v_i$ is elected leader for slot $t$ but is inactive at proposal time, it does not propose in that slot. 

\subsection{Security}
\label{sec:security}

We treat $\lambda$ and $\kappa$ as security parameters for cryptographic primitives and the protocol, respectively, assuming $\kappa > 1$.  
We consider a finite time horizon $\Tconf$, polynomial in $\kappa$.  
An event occurs \emph{with overwhelming probability} (\wop) if it fails with probability $\negl(\kappa) + \negl(\lambda)$.  
Cryptographic guarantees hold with probability $\negl(\lambda)$, though we omit this explicitly in later sections.  

\begin{definition}[Safe protocol]
  \label{def:safety}
 We say that a protocol outputting a confirmed chain $\Chain$
 is \emph{safe} after time $\Tafter$ in executions $\mathcal{E}$, if and only if for any execution in $\mathcal{E}$, two rounds $r, r' \geq \Tafter$, and any two honest validators $v_i$ and $v_j$ (possibly $i=j$) at rounds $r$ and $r'$ respectively, either $\Chain_i^r \preceq \Chain_{j}^{r'}$ or $\Chain_j^{r'} \preceq \Chain_i^r$. If $\Tafter = 0$ and $\mathcal{E}$ includes all partially synchronous executions, then we say that a protocol is \emph{always safe}.
\end{definition}

\begin{definition}[Live protocol]
  \label{def:liveness}
 {We say that a protocol outputting a confirmed chain $\Chain$ is \emph{live} after time $\Tafter$ in executions $\mathcal{E}$, and has confirmation time $\Tconf$, if and only if for any execution in $\mathcal{E}$, any rounds $r \geq \Tafter$ and $r_i \geq r+\Tconf$, any transaction $\tx$ in the transaction pool at time $r$, and any validator $v_i$ active in round $r_i$,  $\tx \in \Chain^{r_i}_i$.} If $\Tafter = 0$ and $\mathcal{E}$ includes all partially synchronous executions, then we say that a protocol is \emph{always live}.
\end{definition}

\begin{definition}[Secure protocol~\cite{goldfish}]
 \label{def:security}
We say that a protocol outputting a confirmed chain $\Chain$ is \emph{secure} after time $\Tafter$, and has confirmation time $\Tconf$, if it is safe after time $\Tafter$ and live after time $\Tafter$ with confirmation time $\Tconf$. A protocol is always secure if it always both safe and live
\end{definition}

We now recall the notions of \emph{Dynamic Availability} and \emph{Reorg Resilience} from~\cite{rlmd}.  

\begin{definition}[Dynamic Availability]
 \label{def:dyn-ava}
A protocol is \emph{dynamically available} after time $\Tdyn$ if and only if it is secure after $\Tdyn$ with confirmation time $\Tconf = O(\kappa)$. We say a protocol is \emph{dynamically available} if it always secure when $\GST = 0$.  
\end{definition}

Note that Dynamic Availability presupposes a certain level of adversarial resilience.  
In the protocol of Section~\ref{sec:ga-based}, Dynamic Availability is attained provided Constraint~\eqref{eq:pi-sleepiness} holds.  

\begin{definition}[Reorg Resilience]
  \label{def:reorg-resilience}
A protocol is \emph{reorg-resilient} after slot $t_{\mathsf{reorg}}$ and time $T_{\mathsf{reorg}}$\footnote{In this work we always have $\Treorg \geq 4\Delta \treorg$.} in executions $\mathcal{E}$ if, for any $\mathcal{E}$-execution, round $r \geq T_{\mathsf{reorg}}$, and validator $v_i$ honest in $r$, any chain proposed in slot $t \geq t_{\mathsf{reorg}}$ by a validator honest in $\proposing{t}$ does not conflict with $\Chain^r_i$. We say a protocol is \emph{reorg-resilient} if this holds after slot $0$ and time $0$.  
\end{definition}

\noindent\textbf{Accountable Safety.}  
In addition to Dynamic Availability, we require accountability in the event of safety violations.  

\begin{definition}[Accountable Safety]
    \label{def:acc-safety}
A protocol has \emph{Accountable Safety} with resilience \(f^{\mathrm{acc}}>0\) if any safety violation admits a cryptographic proof -- constructed from the message transcript -- that identifies at least \(f^{\mathrm{acc}}\) misbehaving validators, while never falsely accusing an honest protocol-following participant.
\end{definition}

By the CAP theorem~\cite{DBLP:journals/corr/abs-2304-14701}, no consensus protocol can ensure both liveness with dynamic participation and safety under partitions.
Thus, Neu, Tas, and Tse~\cite{DBLP:conf/sp/NeuTT21} proposed \emph{ebb-and-flow} protocols, which output two chains: one live under synchrony and dynamic participation, and one safe under asynchrony. 

\begin{definition}[Secure ebb-and-flow protocol]
\sloppy{A secure \emph{ebb-and-flow protocol} outputs an available chain $\chainava$ that is dynamically available\footnote{When we say that a chain produced by protocol $\Pi$ satisfies property $P$, we mean that protocol $\Pi$ itself satisfies $P$.}, and a finalized chain $\chainfin$ that is always safe and live after $\max(\GST,\GAT){+O(\Delta)}$ with $\Tconf = O(\kappa)$. }
Hence, $\chainfin$ is secure after $\max(\GST,\GAT){+O(\Delta)}$ with $\Tconf = O(\kappa)$.  
Moreover, for every honest validator $v_i$ and round $r$, $\chainfin_i^r$ is a prefix of $\chainava_i^r$.
\end{definition}

Our protocol follows the ebb-and-flow approach~\cite{ebbandflow}: $\chainava$ comes from the dynamically available component (Sec.\ref{sec:da-component}), while $\chainfin$ is finalized by a PBFT-style protocol akin to Casper FFG\cite{casper} (Sec.~\ref{sec:ffg}).\\

\noindent\textbf{Fork-Choice Function.}  
A \emph{fork-choice function} $\FC$ is deterministic: given views, a chain, and a slot $t$, it outputs a chain $\chain$. We use the \emph{majority fork-choice function} (\Cref{alg:mfc}), which selects the longest chain supported by a majority of the voting weight.  



\section{The Dynamically Available Component}
\label{sec:da-component}

We begin by presenting the dynamically available component of our protocol.  
In \Cref{sec:recalling-tob}, we recall \TOBSVD{}, summarizing its operation and guarantees.  
In \Cref{sec:revisiting-tob}, we describe our modifications that yield a probabilistically safe, dynamically available consensus protocol.  
We then introduce the notion of fast confirmations and prove its Dynamic Availability (Appendix~\ref{appendix:analysis-prob-ga-fast}), which enables integration into the ebb-and-flow protocol presented in \Cref{sec:ga-based} and fully detailed in \Cref{appendix:analysis-3sf}.  

\subsection{Recalling \TOBSVD}
\label{sec:recalling-tob}

We recall the deterministically safe, dynamically available consensus protocol of D'Amato \emph{et al.}~\cite{streamliningSBFT}.  
We refer the reader to the original work for full technical details, and summarize here its operation, guarantees, and the aspects relevant to our modifications toward a probabilistically safe variant.  
After recalling the original protocol, we show how to introduce fast confirmations and prove Dynamic Availability, which enables its use in our ebb-and-flow protocol.

{\TOBSVD}~\cite{streamliningSBFT} proceeds as follows. During each slot~$t$, a proposal (\(B\)) is made in the first round through a [\textsc{propose}, $B$, $t$, $v_i$] message, and a decision is taken in the third round. During the second round, every active validator \(v_i\) casts a \([\textsc{vote}, \chain, \cdot, t, v_i]\) message for a chain \(\chain\) (The third component in the vote message is omitted here; it will be introduced later when describing the ebb-and-flow protocol.) Importantly, proposals made in
slot~$t$ are decided during the decide phase of slot~$t+1$. All safety
guarantees are fully deterministic.

Let $\V$ be a validator’s view at slot $t$. We define:
\[
\V^{\chain,t} := 
\big\{[\textsc{vote}, \chain', \cdot, \cdot, v_k] \in \V' : 
\V' = \FIL_{\text{lmd}}(\FIL_{1\text{-exp}}(\FIL_{\text{eq}}(\V),t)) 
\ \land\  \chain \preceq \chain'\big\},
\]
where $\FIL_{\text{eq}}$ removes equivocations, $\FIL_{1\text{-exp}}$ discards \emph{expired} votes (i.e., votes older than one slot, which are no longer considered), and $\FIL_{\text{lmd}}$ retains only the latest vote from each validator.

Thus, $\V^{\chain,t}$ is the set of latest, non-expired, non-equivocating votes supporting chains that extend $\chain$.  

The protocol employs a \emph{majority fork-choice function} $\mfc$ (defined formally in \Cref{alg:mfc}): Given views $\V$ and $\V'$, a base chain $\chain^C$, and slot $t$, it returns the longest $\chain \succeq \chain^C$ such that
$|\V^{\chain,t} \cap (\V')^{\chain,t}| > \tfrac{1}{2}\,|\mathsf{S}(\V',t)|,$
where $\mathsf{S}(\V',t)$ is the set of validators appearing with non-expired votes in $\V'$ at slot $t$.  
If no such chain exists, $\mfc$ returns $\chain^C$.

\begin{algo}[h!]
  \caption{$\mfc$, the majority fork-choice function}
  \label{alg:mfc}
  \vbox{
  \small
  \begin{numbertabbing}\reset
    xxxx\=xxxx\=xxxx\=xxxx\=xxxx\=xxxx\=MMMMMMMMMMMMMMMMMMM\=\kill     
  \textbf{function} $\mfc(\V,\V', \chain^C, t)$ \label{}\\
  \> \textbf{return} the longest chain $\chain \succeq \chain^C$ such that\label{}\\
  \>\> $\chain = \chain^C \lor \left|\V^{\chain,t}\cap (\V')^{\chain,t} \right|> \frac{\left|\mathsf{S}(\V',t)\right|}{2}$\label{}\\
  \textbf{function} $\mathsf{S}(\V,t)$\label{}\\
      \> \textbf{return} $ \{v_k: [\textsc{vote}, \cdot, \cdot, \cdot, v_k] \in \V' \land {\V'} = \mathsf{FIL}_{1\text{-exp}}(\V,t)\}$ \label{}\\[-5ex]     
  \end{numbertabbing}
  }
  \end{algo}

The protocol executed by validator $v_i$ proceeds in four phases:
\begin{enumerate}
    \item \textbf{Propose ($4\Delta t$):}  
    If $v_p$ is the proposer for slot $t$, it broadcasts 
$      [\textsc{propose}, \chain_p, t, v_p]
$    through gossip, where $\chain_p \succeq \mfc(\V_p,\V_p,\genesis,t)$ and $\chain_p.p = t$.  
    In practice, $\mfc(\V_i,\V_i,\genesis,t)$ is the parent of $\chain_p$, but we allow the more general case where $\chain_p$ is any extension of this parent with $\chain_p.p = t$.  

    \item \textbf{Vote ($4\Delta t + \Delta$):}  
    Each validator computes
$      \chaincanmfc := \mfc(\V_i^\mathrm{frozen},\V_i,\genesis,t),
$    then gossips a vote message 
     $ [\textsc{vote}, \chain', \cdot, t, v_i],$
    where $\chain'$ is a chain proposed in slot $t$ extending $\chaincanmfc$ with $\chain'.p = t$, or $\chain' = \chaincanmfc$ if no valid proposal exists.  

    \item \textbf{Decide ($4\Delta t + 2\Delta$):}  
    Set 
      $\Chain_i = \mfc(\V''_i,\V_i,\genesis,t), $
    update $\V''_i = \V_i$, and store $\V''_i$.  

    \item \textbf{Freeze ($4\Delta t + 3\Delta$):}  
    Set $\Vfrozen_i = \V_i$ and store it for the next slot.  
\end{enumerate}

\subsection{Probabilistically-Safe Variant} 
\label{sec:revisiting-tob}

A probabilistically safe variant of the protocol just presented is obtained by observing that decisions do not influence validator behavior: they serve only as a confirmation mechanism for deterministic safety.  
In slot $t$, the decision phase stores $\V''$ for use in the next slot but otherwise has no effect.  
Removing this phase simplifies the protocol while preserving probabilistic safety. The main difference -- aside from removing the decide phase and thus reducing 
the slot duration from $4\Delta$ to $3\Delta$ -- is the introduction of a new 
probabilistic confirmation rule. During the vote phase, the protocol designates 
as \emph{confirmed} the $\kappa$-deep prefix of the chain output by $\mfc$. 
The parameter $\kappa$ is chosen so that, within $\kappa$ slots, an honest 
proposer is elected with overwhelming probability, ensuring that consensus 
reliably progresses and that this prefix will not later be reverted. 

\begin{enumerate}
    \item \textbf{Propose ($3\Delta t$):}  
    If $v_i$ is proposer for slot $t$, it gossips  
      $[\textsc{propose}, \chain_p, t, v_i],$
    with $\chain_p \succeq \mfc(\V_i,\V_i,\genesis,t)$ and $\chain_p.p = t$.  

    \item \textbf{Vote ($3\Delta t + \Delta$):}  
    Compute 
      $\chaincanmfc := \mfc(\V_i^\mathrm{frozen},\V_i,\genesis,t),$
    then gossip a vote message $[\textsc{vote}, \chain', \cdot, t, v_i]$,  
    where $\chain'$ is a chain proposed in slot $t$ extending $\chaincanmfc$ with $\chain'.p = t$,  
    or $\chain' = \chaincanmfc$ if no valid proposal exists.  
    Set the confirmed chain as the $\kappa$-deep prefix 
     $ \Chain_i = \left(\chaincanmfc\right)^{\lceil \kappa}$.
    \item \textbf{Freeze ($3\Delta t + 2\Delta$):}  
    Set $\Vfrozen_i = \V_i$, store it, and advance to slot $t+1$.  
\end{enumerate}

Under synchrony, Constraint~\eqref{eq:pi-sleepiness} (adapted to slots of duration $3\Delta$) holds.  
With an honest proposer in slot $t$, the proposed chain $\chain_p$ extends the lock (\ie, $\chain^{3\Delta t + \Delta}_i$), ensuring all honest validators vote for $\chain_p$.  
This alignment of votes yields Reorg Resilience, which in turn provides probabilistic safety under the $\kappa$-deep rule.  
The details are shown in Appendix, in the proof of the final protocol variant.  

\subsection{Probabilistically-Safe Variant with Fast Confirmations} 
\label{sec:prob-tob}

The protocol described above confirms only chains proposed at least $\kappa$ slots in the past.
We now extend it with a mechanism for \emph{fast confirmations}: under certain favorable conditions, a chain can be confirmed in the very slot in which it was proposed.
Concretely, if all validators are awake and more than $\tfrac{2}{3}n$ of them vote for the same chain, that chain is \emph{fast-confirmed} and may be adopted immediately in place of the $\kappa$-deep prefix. 

The advantage of this confirmation rule, compared to the original one, is that 
chains proposed in slot~$t$ can now be confirmed already in slot~$t$, rather 
than waiting until slot~$t+1$ as in the original design. The drawback is that, 
unlike the original rule, this fast confirmation requires full participation in 
the vote phase. However, since the protocol still retains the $\kappa$-deep 
confirmation mechanism, lack of full participation simply causes the protocol 
to fall back to confirming the $\kappa$-deep prefix of the chain output by 
$\mfc$.

This capability (formalized in \Cref{algo:prob-ga-fast}) serves as a building block for our ebb-and-flow protocol.  

Observe that $\kappa$-deep confirmation would suffice, but it yields poor latency, whereas fast confirmation would suffice but makes no progress under poor participation.\\

\noindent\textbf{The Protocol.} \Cref{algo:prob-ga-fast} defines 
$\texttt{fastconfirmsimple}(\V,t) \mapsto (\chain^C, Q^C),$
where $\chain^C$ is a chain supported by more than $\tfrac{2}{3}n$ validators voting in slot $t$ according to view $\V$, and $Q^C$ is a \emph{quorum certificate} consisting of the corresponding \textsc{vote} messages.  
If no such chain exists, it returns $(\genesis,\emptyset)$.  

The protocol also relies on the function $\texttt{Extend}(\chain,t)$: 
If $\chain.p < t$, then $\texttt{Extend}(\chain,t)$ produces a chain 
$\chain'$ of length~$t$ that extends $\chain$ and includes all transactions 
in $\txpool^\proposing{t}$. 
If instead $\chain.p \geq t$, the behavior of the function is left unspecified.

Each honest validator $v_i$ maintains three local objects. 
First, the \emph{message set} $\Vfrozen_i$, which is the snapshot of $\V_i$ 
taken at time $3\Delta$ in each slot. 
Second, the \emph{chain variable} $\chainfrozen_i$, defined as the greatest 
fast-confirmed chain at time $3\Delta$ in slot $t$; this variable is updated 
between $0$ and $\Delta$ of slot $t+1$ using the \textsc{propose} message for slot $t+1$. 
Finally, the \emph{confirmed chain} $\Chain_i$, which is updated at time $\Delta$ 
in each slot and may be updated again at $2\Delta$.

For notation, if $\mathsf{var}$ is a state variable, we write $\mathsf{var}^r$ for its value after round~$r$, or $\Vfrozen[r]_i$ if it already has a superscript.  For readability, we abbreviate $\fastconfirming{t} := 4\Delta t + 2\Delta$ and $\merging{t} := 4\Delta t + 3\Delta$.  

\Cref{algo:prob-ga-fast} proceeds as follows:
\begin{enumerate}
  \item \textbf{Propose ($4\Delta t$):} \sloppy{The proposer $v_p$ broadcasts  
  $[\textsc{propose}, \chain_p, \chain^C, Q^C, \cdot, t, v_p]$,  
  where $\chain_p = \texttt{Extend}(\chaincanmfc, t)$ and $\chaincanmfc = \mfc(\V_p,\V_p,\genesis,t)$. }
  The tuple $(\chain^C, Q^C)$ is computed via $\texttt{fastconfirmsimple}(\V_p,t)$.  

  \item \textbf{Vote ($4\Delta t + \Delta$):} Upon receiving a valid proposal (with $Q^C$ certifying $\chain^C$ and $\chain_p.p = t$), validator $v_i$ updates $\chainfrozen_i$.  
  It then computes $\mfc$ using: (i) $\Vfrozen[\voting{t}]_i$, (ii) $\V^\voting{t}_i$, and (iii) $\chainfrozen_i$.  
  By construction, the voting-phase output is always a prefix of the proposer’s chain (the \emph{Graded Delivery} property~\cite{streamliningSBFT}).  
  If the proposal extends this chain, $v_i$ votes for it; otherwise it votes for its local fork-choice.  
  It also updates $\Chain_i$ to the longer of: (i) the $\kappa$-deep prefix of its $\mfc$ output, and (ii) the previous $\Chain_i$, unless the latter is not a prefix of the former.  

  \item \textbf{Fast Confirm ($\fastconfirming{t}$):} 
  If at least $2/3n$ validators vote for a chain $\chain^\mathrm{cand}$, then $v_i$ fast confirms it by setting $\Chain_i = \chain^\mathrm{cand}$.  

  \item \textbf{Merge ($\merging{t}$):}  
  At round $4\Delta t + 3\Delta$, $v_i$ updates $\Vfrozen_i$ and $\chainfrozen_i$ for use in slot $t+1$.  
\end{enumerate}

\begin{algo}[th!]
  \caption{Probabilistically-safe variant of \TOBSVD{} with fast confirmations - code for $v_i$}
  \label{algo:prob-ga-fast}
  \begin{numbertabbing}\reset
  xxxx\=xxxx\=xxxx\=xxxx\=xxxx\=xxxx\=MMMMMMMMMMMMMMMMMMM\=\kill
    {\textbf{Output}} \label{}\\
    \>{$\Chain_i \gets \genesis$: confirmed chain of validator $v_i$}\label{}\\
    \textbf{State} \label{}\\
    \> $\V_i^\text{frozen}  \gets \{\genesis\}$: snapshot of $\V$ at time $4\Delta t + 3\Delta$ \label{}\\
    \> $\chain^\text{frozen}_i \gets \genesis $: snapshot of the fast-confirmed chain at time $4\Delta t + 3\Delta$ \label{} \\
    \textbf{function} $\texttt{fastconfirmsimple}(\V,t)$\label{}\\
    \> \textbf{let} $\fastcands :=\{\chain \colon |\{v_j\colon \exists \chain' \succeq \chain : \ [\textsc{vote}, \chain', \cdot, t,\cdot] \in \V \}| \geq \frac{2}{3}n\}$\label{}\\
    \> \textbf{if} $\fastcands \neq \emptyset$ \textbf{ then}\label{}\\
    \>\>\textbf{let} $\fastcand := \max\left(\fastcands\right)$\label{}\\
    \>\> \textbf{let} $Q := \{[\textsc{vote},\chain',\cdot,t,\cdot]\in \V : \chain' \succeq \fastcand \}$\label{}\\
    \>\> \textbf{return} $(\fastcand,Q)$ \label{}\\
    \> \textbf{else}\label{}\\
    \>\> \textbf{return} $(\genesis,\emptyset)$\label{}\\    
    \textsc{Propose}\\
    \textbf{at round} $4\Delta t$ \textbf{do} \label{}\\
    \> \textbf{if} $v_i = v_p^t$ \textbf{then} \label{}\\
    \>\> \textbf{let} $ (\chain^C,Q^C) := \texttt{fastconfirmsimple}(\V_i,t-1)$\label{}\\
    \>\> \textbf{let} $\chaincanmfc := \mfc(\V_i, \V_i, \chain^C, t)$ \label{}\\
    \>\> \textbf{let} $ \chain_p := {\mathsf{Extend}(\chaincanmfc,t)}$ \label{line:algga-no-ffg-new-block}\\
    \>\> send message $[\textsc{propose}, \chain_p, \chain^C, Q^C, \cdot, t, v_i]$ through gossip \label{}\\
    \textsc{Vote}\\
    \textbf{at round} $4\Delta t + \Delta$ \textbf{do} \label{}\\
    \> \textbf{let} $\chaincanmfc := \mfc(\Vfrozen_i, \V_i, \chain^\text{frozen}_i, t)$ \label{}\\
    \> $\Chain_i \gets \max(\{\chain \in \{\Chain_i,{(\chaincanmfc)^{\lceil\kappa}}\}: \chain \preceq {\chaincanmfc}\})$\label{line:algga-no-fin-vote-chainava}\\
    \> \textbf{let} $ \chain := $ \textsc{propose}d chain {from slot $t$} extending ${\chaincanmfc}$ and \label{}\\
    \>\> with $\chain.p =t$, if there is one, or ${\chaincanmfc}$ otherwise\label{line:algga-no-fin-vote-comm}\\
    \>  send message $[\textsc{vote}, \chain, \cdot, t, v_i]$ through gossip \label{line:algga-no-fin-vote}\\
    \textsc{Fast Confirm}\\
    \textbf{at round} $4\Delta t + 2\Delta$ \textbf{do} \label{line:algga-no-ffg-on-confirm}\\
    \> \textbf{let} $ (\fastcand,Q) := \texttt{fastconfirmsimple}(\V_i,t)$\label{}\\
    \> \textbf{if} $Q \neq \emptyset$ \textbf{then}\label{line:algga-no-ffg-if-set-chaava-to-bcand}\\
    \>\> $\Chain_i \gets \fastcand$\label{line:algga-no-ffg-set-chaava-to-bcand}\\
    \textsc{Merge}\\
    \textbf{at round} $4\Delta t + 3\Delta$ \textbf{do} \label{}\\
    \>  $\V^\text{frozen}_i \gets \V_i$\label{}\\ 
    \> $(\chainfrozen_i,\cdot) \gets \texttt{fastconfirmsimple}(\V_i,t)$\label{line:algga-no-ffg-merge-ch-frozen}\\
    \\
    \textbf{upon} receiving a message $[\textsc{propose}, {\chain_p}, \chain^C_p, Q^C_p, \cdot, t, v_p]$ \label{}\\
    \>\> \textbf{at any round in} $[4\Delta t, 4\Delta t + \Delta]$ \textbf{do}\label{}\\
      \> \textbf{if} $\mathrm{valid}([\textsc{propose}, {\chain_p}, \chain^C_p, Q^C_p, \cdot, t, v_p] ) \land \chain^\text{frozen}_i \preceq \chain^C_p$ \textbf{then} \label{line:algga-no-ffg-upon}\\  
    \>\> $\chain^\text{frozen}_i \gets \chain^C_p$ \label{algo:3sf-ga-setchainfrozen-to-chainc}\\[-5ex]
  \end{numbertabbing}
\end{algo}
In Appendix~\ref{appendix:analysis-prob-ga-fast}, we prove that \Cref{algo:prob-ga-fast} satisfies   
Reorg Resilience (\Cref{thm:reorg-res-prop-tob}) and
Dynamic Availability (\Cref{thm:dyn-avail-fast-conf-tob}).

We first establish that if all active validators in slot~$t$ send a \textsc{vote} for a chain extending $\chain$, then every validator active in slot~$t{+}1$ will also vote for a chain extending $\chain$, since $\mfc$ necessarily outputs such an extension (\Cref{lem:keep-voting-tob-fast-conf}).  
It follows that if any validator fast confirms a chain $\chain$ in slot~$t$, then all validators active in later slots will output an extension of $\chain$ via $\mfc$ (\Cref{lem:one-fast-confirm-all-vote-fast-conf}).  
A similar guarantee holds for chains proposed by honest validators (\Cref{lem:vote-proposal-fast-conf}).  

These results imply that the confirmed chain of any honest validator at round $r_i$ -- even if the validator is inactive at $r_i$ -- is always a prefix of the $\mfc$ output of any validator active in later rounds (\Cref{lem:ga-confirmed-always-canonical}).  
This property is the key ingredient in proving Reorg Resilience.  

We further show that an honest proposer can include all transactions from the pool in its proposal (\Cref{lem:ga-mfc-proposer-shorter-than-t}).  
Combined with the above lemmas, this yields Dynamic Availability \wop (\Cref{thm:dyn-avail-fast-conf-tob}).  

We prove that fast confirmations are achievable (\Cref{thm:fast-liveness-tob}).  
Although not a formal security property, this result is central to the practical responsiveness of the protocol.  

\section{Finality Component}
\label{sec:ffg}

The finality component of our protocol finalizes, in each slot, a chain that extends the one finalized in the previous slot.  
We introduce the following terminology.  \\

\noindent\textbf{Checkpoints.}
A \emph{checkpoint} is a tuple $\C=(\chain,c)$, where $\chain$ is a chain and
$c$ is the slot in which $\chain$ is proposed for justification.  
By construction $c \ge \chain.p$.
We use $c$ to order checkpoints: $\C \le \C'$ iff $\C.c \le \C'.c$, with
equality only for identical checkpoints.

\noindent\textbf{Justification.}
A set of \textsc{fin-vote}s forms a \emph{supermajority link}
$\C \to \C'$ if at least $2n/3$ validators issue \textsc{fin-vote}s
from $\C$ to $\C'$.
A checkpoint is \emph{justified} if either
(i) $\C=(B_\text{genesis},0)$ or
(ii) it is reachable from $(B_\text{genesis},0)$ by a chain of
supermajority links.
A chain $\chain$ is justified iff some justified checkpoint
$\C$ satisfies $\C.\chain=\chain$.
Given a view $\V$, $\C$ is justified in $\V$ if $\V$ contains the
corresponding supermajority links.
The justified checkpoint with largest slot in $\V$ is
the \emph{greatest justified checkpoint} $\GJ(\V)$, whose chain is the
\emph{greatest justified chain}.  Ties break arbitrarily.

\noindent\textbf{Finalization.}
A checkpoint $\C$ is \emph{finalized} if it is justified and there is a
supermajority link $\C \to \C'$ with $\C'.c = \C.c+1$.
A chain is finalized if it belongs to some finalized checkpoint.
By definition $(B_\text{genesis},0)$ is finalized.
Function $\mathsf{F}(\V,\C)$ 
(\Cref{alg:justification-finalization}) returns true iff $\C$ is finalized in $\V$.
The finalized checkpoint of highest slot is $\GF(\V)$.
A chain $\chain$ is finalized in $\V$ iff $\chain \preceq \GF(\V).\chain$.

\noindent\textbf{Finality Votes.}
Validators issue two types of votes.  
The standard \textsc{vote} message embeds a \textsc{fin-vote}, a tuple
$[\textsc{fin-vote},\C_1,\C_2,v_i]$ linking a source $\C_1$ to a target
$\C_2$.
A \textsc{fin-vote} is \emph{valid} if
$\C_1.c < \C_2.c$ and $\C_1.\chain \preceq \C_2.\chain$.
We write such a vote simply as $\C_1 \to \C_2$.

\noindent\textbf{Equivocations.} A validator $v_i$ \emph{equivocates} if it sends two \textsc{vote} messages in the same slot $t$ for different chains $\chain \ne \chain'$.  

\begin{algo}[t!]
  \caption{Justification and Finalization}
  \label{alg:justification-finalization}
  \vbox{
  \small
  \begin{numbertabbing}\reset
    xxxx\=xxxx\=xxxx\=xxxx\=xxxx\=xxxx\=MMMMMMMMMMMMMMMMMMM\=\kill
  {\textbf{function} $\mathrm{valid}(\C_1 \to \C_2)$ }\label{}\\
  \> {\textbf{return} }{$\land\; \C_1.\chain \preceq \C_2.\chain$}\label{}\\
  \> \hphantom{\textbf{return} }{$\land\; \C_1.c < \C_2.c$}\label{}\\
  \\
  \textbf{function} $\mathsf{J}(\C,\V)$ \label{}\\
    \> \textbf{return} $\lor\; \C = (\genesis,0)$\label{}\\
    \> \hphantom{\textbf{return}} $\lor\; \exists \calS, \M \subseteq \V: \land\; \mathsf{J}(\calS,\V)$\label{}\\
    \> \hphantom{\textbf{return}} $\hphantom{\lor\; \exists \calS, \M \subseteq \V:\,}\land\;\mathrm{valid}(\calS\to\C)$\label{}\\
    \> \hphantom{\textbf{return}} $\hphantom{\lor\; \exists \calS, \M \subseteq \V:\,}\land\; |\{v_k : [\textsc{vote},\cdot,\calS \to \C,\cdot,v_k] \in \M\}|\geq \frac{2}{3}n$
    \label{}\\
    \\
  \textbf{function} $\mathsf{F}(\C,\V)$ \label{}\\
    \> \textbf{return} $\lor\; \C = (\genesis,0)$\label{}\\
    \> \hphantom{\textbf{return}} $\lor\; \land\; \mathsf{J}(\C,\V)$\label{}\\
    \> \hphantom{\textbf{return}} $\hphantom{\lor\;}\land\; \exists \T, \M \subseteq \V: \land\; \T.c = \C.c+1$\label{}\\
    \> \hphantom{\textbf{return}} $\hphantom{\lor\;\land\; \exists \T, \M \subseteq \V:} \land\; \mathrm{valid}(\C\to\T)$\label{}\\
    \> \hphantom{\textbf{return}} $\hphantom{\lor\;\land\; \exists \T, \M \subseteq \V:} \land\; |\{v_k : [\textsc{vote},\cdot,\C \to \T,\cdot,v_k] \in \M\}|\geq \frac{2}{3}n$\label{}\\[-5ex]      
  \end{numbertabbing}
  }
\end{algo}

A validator \(v_i\) is subject to slashing (as introduced in Section~\ref{sec:model}) for two \emph{distinct} \textsc{fin-vote}s \(\C_1 \to \C_2\) and \(\C_3 \to \C_4\) if either of the following conditions holds: {\(\mathbf{E_1}\) (Double voting)} if \(\C_2.c = \C_4.c\), implying that a validator must not send distinct \textsc{fin-vote}s for the same checkpoint slot; or {\(\mathbf{E_2}\) (Surround voting)} if $\mathcal{C}_3.c < \mathcal{C}_1.c < \mathcal{C}_2.c < \mathcal{C}_4.c$. \\

In Appendix~\ref{appendix:proofs}, we prove that if a checkpoint $(\chain_j, c_j)$ is justified while another checkpoint $(\chain_f, c_f)$ is already finalized with $c_j \ge c_f$, then either $\chain_j$ extends $\chain_f$, or at least $n/3$ validators must have violated one of the two slashing conditions (\Cref{lem:accountable-jutification}).
In particular, we prove that if two conflicting chains are ever finalized, then at least $n/3$ validators must have misbehaved (Lemma~\ref{lem:accountable-safety}). 

\section{\protocol}
\label{sec:ga-based}

\begin{figure}
    \centering
    \includegraphics[width=1\linewidth]{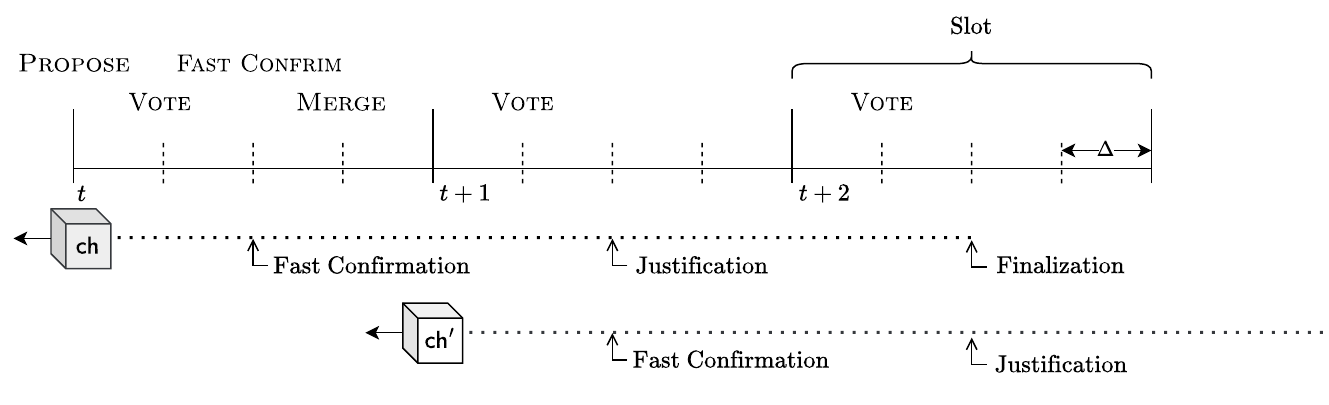}
    \caption{An overview of the slot structure is shown, highlighting phases relevant to the fast confirmation and finalization of a proposal. A proposal $\chain$ is made in the first slot $t$ (top left) and becomes fast-confirmed within the same slot. It is then justified in slot $t+1$ and finalized in slot $t+2$. Similarly, the next proposal $\chain’$ made in slot $t+1$ is fast-confirmed once $\chain$ is justified, and it is justified once $\chain$ is finalized. This illustrates how the confirmation and finalization process is pipelined across consecutive slots.
}
    \label{fig:overview}
\end{figure}

\begin{algo}[h!]
\caption{State and output variables of \protocol (Algorithm~\ref{alg:3sf-tob-noga})}
\label{alg:state}
\begin{numbertabbing}\reset
xxxx\=xxxx\=xxxx\=xxxx\=xxxx\=xxxx\=MMMMMMMMMMMMMMMMMMM\=\kill
  {\textbf{Output}} \label{}\\
  \> {$\chainava_i \gets \genesis$: available chain}\label{}\\
  \> {$\chainfin_i \gets \genesis$: finalized chain}\label{line:algotb-set-chfin-init}\\
  \textbf{State} \label{}\\
  \> $\V_i^\text{frozen}  \gets \{\genesis\}$: snapshot of $\V$ at time $4\Delta t + 3\Delta$ \label{}\\
  \> $\chain^\text{frozen}_i \gets \genesis $: snapshot of the fast-confirmed chain at time $4\Delta t + 3\Delta$ \label{} \\
  \> $\GJ_i^\text{frozen} \gets (\genesis, 0)$: latest frozen greatest justified checkpoint \label{}\\[-5ex]
\end{numbertabbing}
\end{algo}

\begin{algo}[h!]
\caption{\protocol~-- code for validator $v_i$}
\label{alg:3sf-tob-noga}
\begin{numbertabbing}\reset
xxxx\=xxxx\=xxxx\=xxxx\=xxxx\=xxxx\=MMMMMMMMMMMMMMMMMMM\=\kill
  \textbf{function} $\texttt{fastconfirm}(\V,t)$\label{}\\
  \> \textbf{let} $ (\chain^C, Q) := \texttt{fastconfirmsimple}(\V, t)$\label{}\\
  \> \textbf{if} $\chain^C \succeq \GJ(\V).\chain$ \textbf{then}  \textbf{return} $(\chain^C, Q)$\label{}\\
  \> \textbf{else} \textbf{return} $(\GJ(\V).\chain,\emptyset)$\label{}\\
  \textsc{Propose}\\
  \textbf{at round} $4\Delta t$ \textbf{do} \label{}\\
  \> \textbf{if} $v_i = v_p^t$ \textbf{then} \label{}\\
  \>\> \textbf{let} $ (\chain^C,Q^C) := \texttt{fastconfirm}(\V_i,t-1)$\label{}\\
  \>\> \textbf{let} $\chaincanmfc := \mfc(\V_i, \V_i, \chain^C, t)$ and $ \chain_p := {\mathsf{Extend}(\chaincanmfc, t)}$ \label{}\\
  \>\> send message [\textsc{propose}, $\chain_p$, $\chain^C$, $Q^C$, $\GJ(\V_i)$,  $t$, $v_i$] through gossip \label{}\\
  \textsc{Vote}\\
  \textbf{at round} $4\Delta t + \Delta$ \textbf{do} \label{}\\
  \> \textbf{let} $\chaincanmfc := \mfc(\Vfrozen_i, \V_i, \chain^\text{frozen}_i, t)$ \label{line:algtob-set-mfc}\\
  \> $\chainava_i \gets \max(\{\chain \in \{\chainava_i,{(\chaincanmfc)^{\lceil\kappa}},\GJ_i^\text{frozen}.\chain\}: \chain \preceq {\chaincanmfc}\})$\label{line:algtob-vote-chainava}\\
  \> {$\chainfin_i \gets \max(\{\chain \colon \chain \preceq \chainava_i \land \chain \preceq \GF(\V_i).\chain\})$}\label{line:algotb-set-chfin-vote} \\
  \> {\textbf{if} $\GJ_i^\text{frozen}.c = t-1$} \label{line:algtob-set-target-checkpoint-if}\\
  \>\> \textbf{let} $ \T := (\chainava_i,t)$\label{line:algtob-set-target-checkpoint} \\
  \> {\textbf{else}} \label{}\\
  \>\> {\textbf{let} $ \T := (\GJ_i^\text{frozen}.\chain,t)$} \label{line:algtob-rejustification}\\
  \> \textbf{let} $ \chain := $ \textsc{propose}d chain {from slot $t$} extending ${\chaincanmfc}$ \label{}\\
  \>\> and with $\chain.p =t$, if there is one, or ${\chaincanmfc}$ otherwise\label{line:algtob-vote-comm}\\  
  \>  send message [\textsc{vote}, $\chain$, $\GJ_i^\text{frozen} \to \T$, $t$, $v_i$] through gossip \label{line:algtob-vote}\\
  \textsc{Fast Confirm}\\
  \textbf{at round} $4\Delta t + 2\Delta$ \textbf{do} \label{line:algotb-at-confirm}\\
  \> \textbf{let} $ (\fastcand,\cdot) := \texttt{fastconfirm}(\V_i,t)$\label{line:algtob-set-fastcand-fconf}\\
  \> \textbf{if} $\chainava_i \nsucceq \fastcand$ \textbf{then} $\chainava_i \gets \fastcand$\label{line:algtob-set-chaava-to-bcand}\\
  \> {$\chainfin_i \gets \GF(\V_i).\chain$}\label{line:algotb-set-chfin-fast}\\  
  \textsc{Merge}\\
  \textbf{at round} $4\Delta t + 3\Delta$ \textbf{do} \label{}\\
  \>  $\V^\text{frozen}_i \gets \V_i$\label{}\\ 
  \> $(\chainfrozen_i,\cdot) \gets \texttt{fastconfirm}(\V_i,t)$\label{line:algtob-merge-ch-frozen}\\
  \> $\GJfrozen_i \gets \GJ(\V_i)$\label{line:algtob-merge-gj-frozen} \\
  \\
  \textbf{upon} receiving a message
    $[\textsc{propose}, \chain_p, \chain^C_p, Q^C_p, \GJ_p, t, v_p]$ \label{}\\
    \>\> \textbf{at any round in} $[4\Delta t, 4\Delta t + \Delta]$ \textbf{do} \label{line:algtob-upon}\\  
  \> \textbf{when round} $4\Delta t + \Delta$ \textbf{do}\label{}\\
  \>\> \textbf{if} $\text{valid}([\textsc{propose},\chain_p, \chain^C_p, Q^C_p, \GJ_p, t, v_p]) \land \mathsf{J}(\GJ_p,\V_i)  \land$ \label{}\\
  \>\>\>\> $\GJ_p \ge \GJ_i^\text{frozen}$ \textbf{then} \label{line:algtob-prop-if}\\
  \>\>\> $\GJ_i^\text{frozen} \gets \GJ_p$ \label{line:algtob-prop-merge-gj}\\
  \>\>\> \textbf{if} $\chain^\text{frozen}_i \not \succeq \GJ_p.\chain$ \textbf{then}\label{line:algtob-prop-check-chfrozen-gj}\\
  \>\>\>\> $\chain^\text{frozen}_i \gets \GJ_p.\chain$ \label{line:algtob-prop-set-ch-frozen-to-gj}\\
  \>\>\> \textbf{if} $\chain^\text{frozen}_i \preceq \chain^C_p$ \textbf{then}
    \label{line:algtob-if-setchainfrozen-to-chainc}\\
  \>\>\>\> $\chain^\text{frozen}_i \gets \chain^C_p$ \label{algo:3sf-noga-setchainfrozen-to-chainc}\\[-5ex]
\end{numbertabbing}
\end{algo}

We now present \protocol (\Cref{alg:3sf-tob-noga}), which integrates the dynamically available, reorg-resilient protocol of \Cref{algo:prob-ga-fast} with the finality component from \Cref{sec:ffg}, yielding a secure ebb-and-flow protocol.  
We describe \protocol by highlighting the main differences from \Cref{algo:prob-ga-fast}. \Cref{fig:overview} presents an overview of the happy path execution and corresponding slot structure.

The function $\texttt{fastconfirm}(\V,t)$ extends $\texttt{fastconfirmsimple}(\V,t)$ from \Cref{algo:prob-ga-fast}: it returns the same output if the fast-confirmed chain extends $\GJ(\V).\chain$, and otherwise returns $(\GJ(\V).\chain,\emptyset)$.  

Compared to \Cref{algo:prob-ga-fast}, validators maintain one additional variable: 
the \emph{checkpoint variable} $\GJfrozen_i$, defined as the greatest justified 
checkpoint in $v_i$’s view at time $3\Delta$ of slot $t{-}1$. 
This variable is updated during the interval $[0,\Delta]$ of slot $t$ based on 
the \textsc{propose} message.

\protocol outputs two chains:  
\begin{enumerate}
  \item \textbf{Available chain $\chainava_i$:} corresponds to the confirmed chain $\Chain_i$ of \Cref{algo:prob-ga-fast}.  
  \item \textbf{Finalized chain $\chainfin_i$:} set at fast confirmation rounds to $\GF(\V_i).\chain$, and at vote rounds to the longest chain $\chain$ such that $\chainfin_i \preceq \chain \preceq \GF(\V_i).\chain$ and $\chainfin_i \preceq \chainava_i$.  
\end{enumerate}

Relative to \Cref{algo:prob-ga-fast}, the main changes are:  
\begin{enumerate}
    \item \textbf{Propose:} The proposer for slot $t$ includes in its \textsc{propose} message the greatest justified checkpoint in its current view.  
    \item \textbf{Vote:} A proposal $[\textsc{propose}, \chain_p, \chain^C, Q^C, \GJ_p, t, v_p]$ is valid only if $\GJ_p$ is justified in $v_i$’s view.  
    To allow \textsc{fin-vote}s to accumulate, validators delay processing such proposals until $\voting{t}$.  
    If valid and $\GJ_p \succ \GJfrozen_i$, then $v_i$ updates both $\GJfrozen_i$ and $\chainfrozen_i$ (ensuring $\chainfrozen_i \succeq \GJfrozen_i$).  
    The available chain $\chainava_i$ is then set to the maximum among: (i) its previous value, (ii) the $\kappa$-deep prefix of the current fork-choice, and (iii) $\GJfrozen_i.\chain$, discarding any option not extending the fork-choice.  
    Each validator also includes a \textsc{fin-vote} in its \textsc{vote} message, with source $\GJfrozen_i$ and target depending on its slot: if $\GJfrozen_i$ is from slot $t{-}1$, the target is $\chainava_i$; otherwise, the target is $\GJfrozen_i.\chain$ lifted to slot $t$.  
    \item \textbf{Fast Confirm:} If $\texttt{fastconfirm}(\V_i,t)$ returns a candidate $\chain^\mathrm{cand}$, then $v_i$ sets $\chainava_i := \chain^\mathrm{cand}$.  
    \item \textbf{Merge:} At round $4\Delta t + 3\Delta$, $v_i$ updates $\GJfrozen_i$ to the greatest justified checkpoint in its current view.  
\end{enumerate}

For $\GST=0$, Appendix~\ref{appendix:analysis-3sf} shows that $\chainava$ in \protocol is  
dynamically available \wop (Theorem~\ref{thm:dyn-avail-fast-conf-tob}) and reorg-resilient (Theorem~\ref{thm:reorg-res-prop-tob}).

At a high level, assuming $\GST=0$, $f<n/3$, and Constraint~\eqref{eq:pi-sleepiness}, Appendix~\ref{sec:analysis-tob-sync} shows that integrating finality does not alter the core behavior of \Cref{algo:prob-ga-fast}.  
Concretely, we prove that: (i) honest validators send identical messages in both protocols except for the finality component, (ii) their fork-choice outputs coincide, and (iii) every $\chainava$ in \protocol is a confirmed chain of \Cref{algo:prob-ga-fast}.  
Thus, all lemmas and theorems proved in Appendix~\ref{appendix:analysis-prob-ga-fast} for \Cref{algo:prob-ga-fast} also apply to \protocol.  

Finally, in Appendix~\ref{sec:analysis-tob-psync}, we show that after 
time $\max(\GST,\GAT)+4\Delta$, \protocol guarantees both 
accountable safety and liveness. In particular, it satisfies three key 
properties. First, \emph{accountable safety}: honest validators are never 
slashed (Lemma~\ref{lem:never-slashed-3sf-tob-noga}); hence, the existence 
of two conflicting finalized chains implies that at least $\tfrac{n}{3}$ 
validators violated a slashing condition 
(Theorem~\ref{thm:accountable-safety}). Second, \emph{monotonicity}: the 
finalized chain $\chainfin_i$ of any honest validator grows monotonically 
(Lemma~\ref{lem:chfin-ga-always-grows}). Third, \emph{liveness}: after 
$\max(\GST,\GAT)+4\Delta$, if two honest proposers are elected in consecutive 
slots, then the first proposer’s chain is finalized within 2--3 slots 
(Theorem~\ref{thm:liveness-chfin}), depending on whether rejustification is 
needed. We conclude in Theorem~\ref{thm:ga-ebb-and-flow} that \protocol satisfies the properties of a secure ebb-and-flow protocol \wop  \\

\noindent\textbf{Communication Complexity} \protocol has expected communication complexity $\mathcal{O}(Ln^3)$, where $L$ is the block size.  
This matches the bound of the underlying dynamically available protocol \TOBSVD{}, since honest validators forward all received messages in each round, contributing to the $\mathcal{O}(Ln^3)$ cost.  
The finality gadget introduces no additional overhead, as it is embedded in the \textsc{vote} message at round $4\Delta t + \Delta$.  
Thus, the overall complexity remains identical to that of \TOBSVD{}.

\section{Conclusion}
\label{sec:conclusion}

We presented \protocol, a secure ebb-and-flow consensus protocol that achieves security and efficiency with only a single voting phase per slot.
Building on Momose and Ren~\cite{DBLP:conf/ccs/Momose022} and subsequent work, we adapted \TOBSVD{} as a dynamically available foundation and simplified it to probabilistic safety, yielding a protocol that is both simpler and more practical while retaining its guarantees.
Combined with a finality protocol, \protocol finalizes chains within two or three slots under optimistic conditions, while preserving the communication complexity of its underlying dynamically available protocol.

\bibliography{main}
\bibliographystyle{plainurl}

\appendix
\section{Analysis of Algorithm~\ref{algo:prob-ga-fast}}
\label{appendix:analysis-prob-ga-fast}

\begin{lemma}
\label{lem:2honestproposers}
Consider a sequence of \( \kappa \) independent slots, where in each slot an honest validator is elected as proposer with probability at least $p = \frac{h_0}{n},$
where \( h_0 \) is the number of active honest validators and \( n \) is the total number of validators. For a finite time horizon \( \Tconf \), polynomial in \( \kappa \), if \( p \) is non-negligible, then with overwhelming probability, i.e., with probability at least \( 1 - \negl(\kappa) - \negl(\lambda) \), there exists at least one pair of consecutive slots in which honest proposers are elected.
\end{lemma}

\begin{proof}
Let \( X_i \) denote the event that the \( i \)-th slot has an honest proposer. By assumption, each \( X_i \) occurs independently with probability at least \( p \). Let \( E \) be the event that no two consecutive slots have honest proposers. We conservatively upper bound \( P(E) \) by partitioning the sequence into \( \lfloor \kappa / 2 \rfloor \) disjoint groups of two consecutive slots. In each such pair, the probability that neither of the two adjacent pairs  \( (i, i+1) \) both contain honest proposers is at most \( 1 - p^2 \).

Therefore, the probability of event \( E \) is at most:
\[
P(E) \leq (1 - p^2)^{\lfloor \kappa / 2 \rfloor}.
\]

If \( p \) is non-negligible and \( \kappa \) is treated as a security parameter, then this upper bound becomes negligible in \( \kappa \); i.e., \( P(E) \leq \negl(\kappa) \).

Since cryptographic failures occur with probability \( \negl(\lambda) \) (Section~\ref{sec:security}), it follows that the probability of failing to observe two consecutive honest proposers is at most \( \negl(\kappa) + \negl(\lambda) \).

Hence, with overwhelming probability, there exists at least one pair of consecutive slots in which honest proposers are elected.
\end{proof}

For any slot $t$, we define:
\begin{itemize}
    \item \sloppy{$\mfcpropose{t}_i := \mfc(\V^{\proposing{t}}_i, \V^{\proposing{t}}_i, \chain^C, t)$, where $(\chain^C, \cdot) = \texttt{fastconfirmsimple}(\V^{\proposing{t}}_i, t{-}1)$;}
    \item $\mfcvote{t}_i := \mfc(\Vfrozen[\voting{t}]_i, \V^{\voting{t}}_i, \chainfrozen[\voting{t}]_i, t)$.
\end{itemize}

\begin{lemma}
    \label{lem:keep-voting-tob-fast-conf}
    If, in slot $t$, all validators in $H_{\voting{t}}$ send \textsc{vote} messages for chains extending chain \(\chain\),
    then, for any validator $v_i\in H_{\voting{(t+1)}}$, $\mfcpropose{(t+1)}_i \succeq \chain$ and $\mfcvote{(t+1)}_i \succeq \chain$, which implies that, in slot $t+1$, all validators in $H_{\voting{(t+1)}}$ send \textsc{vote} messages for chains extending \(\chain\).
\end{lemma}

\begin{proof}
Let $v_i$ be any honest validator in \(H_{\voting{(t+1)}}\).
Due to the joining protocol, this implies that $v_i$ is awake at round $\merging{t} = 4\Delta (t+1)-\Delta$.
Clearly the same holds at time $\proposing{(t+1)}$: $\chain^{\text{frozen},\proposing{(t+1)}}_i = \genesis \lor \chain^{\text{frozen},\proposing{(t+1)}}_i\succeq \chain$.
Now we consider each case separately and for each of them prove that $\mfcpropose{(t+1)}_i \succeq \chain$.
\begin{description}
  \item[Case 1: {$\chainfrozen[\proposing{t+1}]_i = \genesis$}.]
  Note that due to the Lemma's and synchrony assumptions, all the \textsc{vote} messages sent by validators in \(H_{\voting{t}} \setminus A_{\voting{(t+1)}}\) are included in $\V_i^{\proposing{(t+1)}}$.
   Hence, at time $\proposing{(t+1)}$, \(\left|\V_i^{\chain,t+1} \right| \geq \left|H_{\voting{t}}    \setminus A_{\voting{(t+1)}}\right|\).
  Similarly, $H_{\voting{t}} \setminus A_{\voting{(t+1)}} \subseteq \mathsf{S}(\V_i,t+1)$.
  Also, $\mathsf{S}(\V_i,t+1)\setminus \left(H_{\voting{t}} \setminus A_{\voting{(t+1)}}\right) \subseteq A_{\voting{(t+1)}} \cup \left(H_{\voting{(t-1+1)},\voting{(t-1)}}\setminus H_{\voting{t}}\right)$ as any \textsc{vote} in $\mathsf{S}(\V_i,t+1)$ that is not from a validator in $H_{\voting{t}} \setminus A_{\voting{(t+1)}}$ must be either from a validator Byzantine at time $\voting{(t+1)}$ or from a validator active at some point between $\voting{(t-1+1)}$ and $\voting{(t-1)}$ but not active at time $\voting{t}$.
  Then, we have 
 
\begin{align*}
\left|\mathsf{S}(\V_i,t+1)\right| 
&= \left|H_{\voting{t}} \setminus A_{\voting{(t+1)}} \right| 
+ \left|\mathsf{S}(\V_i,t+1) \setminus \left(H_{\voting{t}} \setminus A_{\voting{(t+1)}}\right)\right| \\
&\leq \left|H_{\voting{t}} \setminus A_{\voting{(t+1)}} \right| 
+ \left|A_{\voting{(t+1)}} \cup \left(H_{\voting{(t-1+1)},\voting{(t-1)}} \setminus H_{\voting{t}}\right)\right| \\
&< 2\left|H_{\voting{t}} \setminus A_{\voting{(t+1)}} \right|
\end{align*}

  where the last inequality comes from \Cref{eq:pi-sleepiness}.

  {Hence, at time $\proposing{(t+1)}$, \(\left|\V_i^{\chain,t+1} \right|\)\(\geq \left|H_{\voting{t}} \setminus A_{\voting{(t+1)}}\right|  > \frac{\left|\mathsf{S}(\V_i,t+1)\right|}{2}\).}
  
  Therefore, $\mfcpropose{(t+1)}_i \succeq \chain$.  
  \item[Case 2: {$\chainfrozen[\proposing{t+1}]_i \succeq \chain$}.] 
  \sloppy{{Given that by definition, $\mfcpropose{(t+1)}_i \succeq \chainfrozen[\proposing{t+1}]_i $, we have that $\mfcpropose{(t+1)}_i \succeq \chain$.}}
\end{description} 

\sloppy{Now let us move to proving that $\mfcvote{(t+1)}_i  \succeq \chain$.
Note that if a $[\textsc{propose}, B, \chain_p^C, Q^C_p, t+1, v_p]$ message is valid, then $\chain_p^C = \genesis \lor \chain_p^C \succeq \chain$.
\Cref{algo:3sf-ga-setchainfrozen-to-chainc} implies that $\chainfrozen[\voting{(t+1)}]_i = \chainfrozen[\merging{t}]_i \lor \chainfrozen[\voting{(t+1)}]_i = \chain_p^C$.}
Then, from $\chain^{\text{frozen},\merging{t}}_i = \genesis \lor \chain^{\text{frozen},\merging{t}}_i\succeq \chain$ and \Cref{algo:3sf-ga-setchainfrozen-to-chainc}, it follows that, 
at time $\voting{(t+1)}$, $\chain^{\text{frozen},\voting{(t+1)}}_i = \genesis \lor \chain^{\text{frozen},\voting{(t+1)}}_i \succeq \chain$.
Now we consider each case separately. 

\begin{description}
  \item[Case 1: {$\chainfrozen[\voting{t+1}]_i = \genesis$}.]
  Note that due to the Lemma's and synchrony assumptions, all the \textsc{vote} messages sent by validators in \(H_{\voting{t}} \setminus A_{\voting{(t+1)}}\) are included in $\V^{\text{frozen},\voting{(t+1)}}_i\cap \V_i^{\voting{(t+1)}}$.
  From here the proof is effectively the same as the one for Case 1 for $\mfcpropose{(t+1)}_i$, just with $\V_i^{\proposing{(t+1)}}$ replaced by $\V^{\text{frozen},\voting{(t+1)}}_i\cap \V_i^{\voting{(t+1)}}$.
  \item[Case 2: {$\chainfrozen[\voting{t+1}]_i \succeq \chain$}.] 
  Similarly to proof of Case 2 for $\mfcpropose{(t+1)}_i$, given that by definition, $\mfcvote{(t+1)}_i \succeq \chainfrozen[\voting{t+1}]_i $, we have that $\mfcvote{(t+1)}_i \succeq \chain$.l
\end{description} 

From $\mfcvote{(t+1)}_i  \succeq \chain$ follows that, in slot $t+1$,  $v_i$ sends a \textsc{vote} message for a chain extending~$\chain$.
\end{proof}

\begin{lemma}
    \label{lem:one-fast-confirm-all-vote-fast-conf}
    \sloppy{If an honest validator fast confirms a chain $\chain$ in slot $t$ ($\ie$, there exists $Q$ such that $(\chain, Q) = \texttt{fastconfirmsimple}(\V^{\fastconfirming{t}}_i,t) \land Q \neq \emptyset$), then, for any slot $t'>t$ and validator $v_i \in H_{\voting{(t')}}$, $\mfcpropose{t'}_i \succeq \chain$ and $\mfcvote{t'}_i \succeq \chain$, which implies that, all validators in $H_{\voting{(t')}}$ send a \textsc{vote} message for a chain extending $\chain$.}
\end{lemma}

\begin{proof}
The proof is by induction on $t'$.

\begin{description}
  \item[Base Case: $t'=t+1$.] 
  Assume that at round $\fastconfirming{t}$, an honest validator in \(H_{\voting{t}}\) fast confirms a chain $\chain$.
  Given that we assume $f < \frac{n}{3}$, conflicting quorum certificates cannot form in the same slot.
  Then, due to the joining protocol and the synchrony assumption, at time $\merging{t} = 4\Delta (t+1) - \Delta$, $\chain^\text{frozen, \merging{t}}_i = \chain$ for any validator $v_i \in H_{\voting{(t+1)}}$.

  Similarly, if a $[\textsc{propose}, \chain_p, \chain_p^C, Q^C_p, t+1, v_p]$ message is valid, then $\chain_p^C = \genesis \lor \chain_p^C = \chain$.
  This and \Crefrange{line:algga-no-ffg-upon}{algo:3sf-ga-setchainfrozen-to-chainc} imply that, at time $\voting{(t+1)}$, $\chain^{\text{frozen},\voting{(t+1)}}_i =  \chain$.

  Hence, because $\mfcvote{(t+1)}_i \succeq \chainfrozen[\voting{t+1}]_i $, $\mfcvote{(t+1)}_i \succeq \chain$, implying that all validators in $H_{\voting{(t+1)}}$ send \textsc{vote} messages extending chain $\chain$.
  \item[Inductive Step: $t' > t+1.$] 
  Here we can just apply \Cref{lem:keep-voting-tob-fast-conf} to conclude the proof.\qed
\end{description}
\end{proof}

\begin{lemma}
\label{lem:vote-proposal-fast-conf}
Let $t$ be a slot with an honest proposer $v_p$ and assume that~$v_p$ sends a $[\textsc{propose}, \chain_p, \chain^C_p, Q^C_p, t, v_p]$ message. 
Then, for any slot $t'\geq t$, all validators in $H_{\voting{(t')}}$ send a \textsc{vote} message for a chain extending $\chain_p$.
Additionally, for any slot $t'' > t$ and any validator $v_i\in H_{\voting{(t'')}}$, $\mfcpropose{t''}_i\succeq \chain_p$ and $\mfcvote{t''}_i\succeq \chain_p$.
\end{lemma}

\begin{proof}
The proof is by induction on $t'$.

\begin{description}
  \item[Base Case: $t' = t$.] 
  \sloppy{Suppose that in slot \(t\), an honest proposer $v_p$ sends a $[\textsc{propose}, \chain_p, \chain^C_p, Q^C_p t, v_p]$ message.}
  Consider an honest validator \(v_i \in H_{\voting{t}}\).
  Note that due to the synchrony assumption, $\V^{\merging{(t-1)}}_i \subseteq \V^{\proposing{t}}_p$.
  Given that we assume $f < \frac{n}{3}$, this further implies that $\chain^{\text{frozen},\merging{(t-1)}}_i \neq \genesis \implies \chain^C_p  = \chain^{\text{frozen},\merging{(t-1)}}_i $.
  Hence, clearly $\chain^C_p \succeq \chain^{\text{frozen},\merging{(t-1)}}_i$.
  Therefore, due to \Crefrange{line:algga-no-ffg-upon}{algo:3sf-ga-setchainfrozen-to-chainc}, $\chain^{\text{frozen},\voting{t}}_i = \chain^C_p$.
  
  We know that either \(\left|\left(\Vfrozen_i\right)^{\mfcvote{t}_i,t} \cap \V_i^{\mfcvote{t}_i,t} \right| > \frac{\left|\mathsf{S}(\V_i,t)\right|}{2}\) or $\mfcvote{t}_i = \chain^\text{frozen}_i$.
  Let us consider each case separately.

  \begin{description}
    \item[Case 1: \(\left|\left(\Vfrozen\right)^{\mfcvote{t}_i,t} \cap \V_i^{\mfcvote{t}_i,t} \right| > \frac{\left|\mathsf{S}(\V_i,t)\right|}{2}\).]  
    By the Graded Delivery property~\cite{streamliningSBFT}, this implies that at time $\proposing{t}$, $\left|\V_p^{\mfcvote{t}_i,t}\right|>\frac{\left|\mathsf{S}(\V_p,t)\right|}{2}$ meaning that, due to \Cref{line:algga-no-ffg-new-block}, $\chain_p\succeq \mfcvote{t}_i$ and hence, due to \Cref{line:algga-no-fin-vote,line:algga-no-fin-vote-comm}, in slot $t$, $v_i$ sends a \textsc{vote} message for $\chain_p$.
  
    \item[Case 2: $\mfcvote{t}_i = \chain^\text{frozen}_i$.]
    Due to \Cref{line:algga-no-fin-vote,line:algga-no-fin-vote-comm}, $v_i$ still sends a \textsc{vote} for $\chain_p$ as $\chain_p \succeq \chain^C_p = \chain^\text{frozen}_i = \mfcvote{t}_i$.
  \end{description}
  \item[Inductive Step: $t' > t$.]
  Here we can just apply \Cref{lem:keep-voting-tob-fast-conf} to conclude the proof.\qed
\end{description}
\end{proof}

\begin{lemma}
  \label{lem:ga-confirmed-always-canonical}
  Let $r_i$ be any round and $r_j$ be any round such that $r_j\ge r_i$ and $r_j \in \{\proposing{\slot(r_j)}, \voting{\slot(r_j)}\}$. Then, for any validator~$v_i$ honest in round $r_i$ and any validator $v_j \in H_{r_j}$, $\Chain^{r_i}_i \preceq \mfc^{r_j}_j$.
\end{lemma}

\begin{proof}

We proceed by contradiction.
\sloppy{Let $r_i$ be the smallest round such that there exist two honest validators $v_i$ and $v_j$, and round $r_j$ such that $r_j\ge r_i$ and $r_j \in \{\proposing{\slot(r_j)}, \voting{\slot(r_j)}\}$ and 
$\Chain^{r_i}_i \npreceq \mfc^{r_j}_j$, that is, $r_i$ is the first round where the chain confirmed by an honest validator conflicts with the output of $\mfc$ of (another) honest validator at a propose or vote round $r_j \geq r_i$.}
Given the minimality of $r_i$, $\Chain^{r_i-1}_i \neq \Chain^{r_i}_i$ which then implies that $v_i \in H_{r_i}$.
This can only happen if $r_i$ is either a voting or a fast confirmation round.
Let $t_i= \mathrm{slot}(r_i)$ and proceed by analyzing each case separately. 
\begin{description}
\item[Case 1: $r_i$ is a vote round.] 
Due to \Cref{line:algga-no-fin-vote-chainava}, $\Chain^{r_i}_i \succeq \left(\mfcvote{t_i}_i\right)^{\lceil \kappa}$.
Let us now consider two sub cases.
  \begin{description}
    \item[Case 1.1: $\Chain^{r_i}_i = \left(\mfcvote{t_i}_i\right)^{\lceil \kappa}$.] 
    With overwhelming probability (\Cref{lem:2honestproposers}), there exists at least one slot \(t_p\) in the interval \([t_i - \kappa, t_i)\) with an honest proposer $v_p$.
    Let $\chain_p$ be the chain \textsc{propose}d by $v_p$ in slot $t_p$.
    Given that $t_p < t_i$, \Cref{lem:vote-proposal-fast-conf} implies that 
    $\mfcvote{t_i}_i \succeq \chain_p$.
    Then, because $t_p\geq t_i-\kappa$, we have that $\chain_p \succeq  \left(\mfcvote{t_i}_i\right)^{\lceil \kappa} = \Chain^{r_i}_i$.
    Because $t_p < \slot(r_j)$, \Cref{lem:vote-proposal-fast-conf} also implies that $\chain^{r_j}_j \succeq \chain_p \succeq \Chain^{r_i}_i$ leading to a contradiction.
    \item[Case 1.2: $\Chain^{r_i}_i \succ \left(\mfcvote{t_i}_i\right)^{\lceil \kappa}$.] 
    This case implies that $\Chain^{r_i}_i  = \Chain^{r_i-1}_i$.
    From the minimality of $r_i$ we reach a contradiction.
  \end{description}
  \item[Case 2: $r_i$ is a fast confirmation round.]
  Note that this implies that $t_i < \slot(r_j)$.
  Because of the minimality of $r_i$, we only need to consider the case that $(\Chain^{r_i}_i, Q) = \texttt{fastconfirmsimple}(\V^{r_i}_i,t_i) \land Q \neq \emptyset$.
  Therefore, we can apply \Cref{lem:one-fast-confirm-all-vote-fast-conf} to conclude that $\chain^{r_j}_j \succeq \Chain^{r_i}_i$ reaching a contradiction.\qed
\end{description}
\end{proof}

\begin{theorem}[Reorg Resilience]\label{thm:reorg-res-prop-tob}
  \Cref{algo:prob-ga-fast} is reorg-resilient. 
\end{theorem}

\begin{proof}
  Take a slot $t_p$ with an honest proposer $v_p$ who sends a \textsc{propose} message for chain $\chain_p$.
  Take also any round $r_i$ and validator $v_i$ honest in round $r_i$.
  Now let $t_j$ be any slot such that $t_j > \max(t_p,\slot(r_i))$.
  Pick any validator $v_j \in H_{\voting{(t_j)}}$.
  \Cref{lem:vote-proposal-fast-conf} implies that $\mfcvote{t_j}_j \succeq \chain_p$.
  \Cref{lem:ga-confirmed-always-canonical} implies that $\mfcvote{t_j}_j \succeq \Chain^{r_i}_i$.
  Hence, $\chain_p$ does not conflict with $\Chain^{r_i}_i$.
\end{proof}

{
\begin{lemma}\label{lem:ga-mfc-proposer-shorter-than-t}
  For any slot $t$ and  validator $v_i \in H_\proposing{t}$, $\mfcpropose{t}_i.p < t$.
\end{lemma}
\begin{proof}
  Due to the Validity property of Graded Agreement~\cite{streamliningSBFT} and the fact that $\genesis.p <0$,
  it is sufficient to prove that no validator honest in round $\proposing{t}$ has ever sent a \textsc{vote} message for a chain $\chain$ with $\chain.p \geq t$.
  The proof is by induction.
  \begin{description}
    \item[Base Case: $t=0$.] Obvious as no validator in $H_\proposing{t}$ has ever sent any message.
    \item[Inductive Step: $t >0$.] 
    Due to the inductive hypothesis and the fact that honest nodes do not send any \textsc{vote} message in round $\proposing{(t-1)}$, by the Validity property of Graded Agreement, we know that $\mfcvote{t-1}.p < t-1$.
    Then, the proof follows from \Cref{line:algga-no-fin-vote-comm}.\qed
  \end{description}
\end{proof}
}

\begin{theorem}
    \label{thm:dyn-avail-fast-conf-tob}
\Cref{algo:prob-ga-fast} is dynamically available \wop.
\end{theorem}

\begin{proof}
We begin by proving the Liveness of the protocol with a confirmation time of {$T_{\text{conf}} = 8\kappa\Delta + \Delta$}.
We prove liveness by first considering the $\kappa$-deep rule only.
Take a round \( r \) at slot \( t = \text{slot}(r) \), another round {\( r_i \geq r + 8\kappa\Delta + \Delta \geq  4\Delta(t+2\kappa)+\Delta = \voting{(t+2\kappa)}\)}, and an honest validator \( v_i \in H_{r_i} \).
Let $t_i = \mathrm{slot}(r_i)$.
Due to the joining protocol, we know that the first active round for $v_i$, at or after $\voting{(t+2\kappa)}$, is a vote round.
There is a high probability of finding a slot \( t_p  \in [t + 1, t + \kappa]\) hosted by an honest proposer (\Cref{lem:2honestproposers}).
Let $\chain_p$ be the chain \textsc{propose}d by the honest proposer $v_p$ in slot $t_p$.
{Due to \Cref{lem:ga-mfc-proposer-shorter-than-t} and Property 1, we know that $\chain_p$ includes all of the transaction in $\txpool^\proposing{t_p}$.
Given that $\proposing{t_p} \geq r$, $\txpool^r \subseteq \txpool^\proposing{t_p}$, which implies that $\chain_p$ includes all of the transactions in $\txpool^r$.
}
Given that $\mathrm{slot}(r_i)>t_p$, as a consequence of \Cref{lem:vote-proposal-fast-conf}, $\mfcvote{t_i}_i \succeq \chain_p$.
Note that \Cref{line:algga-no-fin-vote-chainava} implies that $\Chain^{r_i}_i \succeq \left(\mfcvote{t_i}_i\right)^{\lceil \kappa,t_i}$.
Then, because \( t_p \leq t + \kappa \leq t_i - \kappa \), \( \chain_p \preceq \left(\mfcvote{t_i}_i\right)^{\lceil \kappa,t_i}\) and hence $\chain_p \preceq \Chain_i^{r_i}${, which implies that $\Chain_i^{r_i}$ includes any transaction in $\txpool^r$}.

We now want to show that fast confirmation does not interfere.
Given that $t_i \geq t_j$, from \Cref{lem:vote-proposal-fast-conf}, we know that any \textsc{vote} sent in slot $t_i$ by honest validators are for chains extending $\chain_p$, given that we assume $f< \frac{n}{3}$, if $r_i$ is a fast confirmation round and  $v_i$  sets $\Chain^{r_i}_i$, then $\Chain^{r_i}_i\succeq\chain_p$ still holds.

We now show safety.
Take any two rounds $r_i$ and $r_j$ and validators $v_i$ and $v_j$ honest in round $r_i$ and $r_j$ respectively.
Now let $t_k$ be any slot such that $t_k > \max(\slot(r_j),\slot(r_i))$.
By Constraint~\eqref{eq:pi-sleepiness} we know that $H_{\voting{(t_k)}}$ is not empty.
So, pick any validator $v_k \in H_{\voting{(t_k)}}$.
\Cref{lem:ga-confirmed-always-canonical} implies that $\mfcvote{t_k}_k \succeq \Chain^{r_i}_i$ and $\mfcvote{t_k}_k \succeq \Chain^{r_j}_j$.
Hence, $\Chain^{r_i}_i$ does not conflict with $\Chain^{r_j}_j$.
\end{proof}

\begin{lemma}[Liveness of fast confirmations]
\label{thm:fast-liveness-tob}
Take a slot $t$ in which \(|H_{\voting{t}}| \geq \frac{2}{3}n\).
If in slot $t$ an honest validator sends a \textsc{propose} message for chain $\chain_p$, then, for any validator $ v_i \in H_{\fastconfirming{t}}$, $\Chain^{\fastconfirming{t}}_i \succeq \chain_p$.
\end{lemma}

\begin{proof}
By assumption we have that \(|H_{\voting{t}}| \geq \frac{2}{3}n\), implying from \Cref{lem:vote-proposal-fast-conf} that every honest validator in \(H_{\voting{t}} \) sends a \textsc{vote} message for $\chain_p$. It follows that any validator $ v_i \in H_{\fastconfirming{t}}$ receives a quorum of \textsc{vote} messages for $\chain_p$. 
Hence, due to \Cref{line:algga-no-ffg-if-set-chaava-to-bcand,line:algga-no-ffg-set-chaava-to-bcand},$\Chain^{\fastconfirming{t}}_i \succeq \chain_p$.
\end{proof}

\section{Analysis of the Finality-Component}
\label{appendix:proofs}

\begin{lemma}
\label{lem:accountable-jutification}
Assume $f \in [0,n]$ and let
$\C_\mathsf{f}=(\chain_\mathsf{f},c_\mathsf{f})$ be a finalized checkpoint and $\C_\mathsf{j}=(\chain_\mathsf{j},c_\mathsf{j})$ be a justified checkpoint with $c_\mathsf{j} \geq c_\mathsf{f}$ {according to any two views}.
Either $\chain_\mathsf{j} \succeq \chain_\mathsf{f}$ or at least $\frac{n}{3}$ validators can be detected to have violated either $\mathbf{E_1}$ or $\mathbf{E_2}$.
\end{lemma}

\begin{proof}
{First, we show} that no checkpoint $\C'=(\chain',c_\mathsf{f})$, with $\chain' \neq \chain_\mathsf{f}$, {can ever be justified.}
If that is not the case, clearly $\geq \frac{n}{3}$ validators are slashable for violating $\mathbf{E_1}$: For the justification of \((\chain_\mathsf{f},c_\mathsf{f})\), it is required to have a supermajority of \textsc{fin-vote}s with the chain of the target checkpoint being $\chain_\mathsf{f}$. Similarly, the justification of \((\chain',c_\mathsf{f})\) {requires to have a supermajority of \textsc{fin-vote}s with the chain of the target checkpoint being $\chain'$}. Given that we need a supermajority set of links for justification, the intersection of the sets of voters contributing to the justification of $\C_\mathsf{f}$ and $\C'$ include at least $\frac{n}{3}$ validators {which have sent two \textsc{fin-vote}s $\calS_\mathsf{f} \to \T_\mathsf{f}$ and $\calS' \to \T'$ with $\T_\mathsf{f}.c = \T'.c = c_\mathsf{f}$ and $\T_\mathsf{f}.\chain \neq \T'.\chain$ thereby violating condition $\mathbf{E_1}$}.

Now, by contradiction, assume that $\chain_\mathsf{j} \not\succeq \chain_\mathsf{f}$ and that there does not exist a set of at least $\frac{n}{3}$ validators that can be detected to have violated either $\mathbf{E_1}$ or $\mathbf{E_2}$.
Let \(c'_\mathsf{j} > c_\mathsf{f}\) be the smallest slot for which a checkpoint \(\C'_\mathsf{j} = (\chain'_\mathsf{j}, c'_\mathsf{j})\) is justified with $\chain'_\mathsf{j} \not\succeq \chain_\mathsf{f}$, \ie, either $\chain'_\mathsf{j}$ conflicts with $\chain_\mathsf{f}$ or $\chain'_\mathsf{j}$ is a strict prefix of $\chain_\mathsf{f}$. Given our assumptions we know that one such a checkpoint exists.

Let $(A_i,c_i) \to (\chain'_\mathsf{j},c'_\mathsf{j})$ and $(\chain_\mathsf{f},c_\mathsf{f}) \to (C,c_\mathsf{f}+1)$ be the \textsc{fin-vote}s involved in the justification of $(\chain'_\mathsf{j},c'_\mathsf{j})$ and in the finalization of $\chain_\mathsf{f}$, respectively, sent by a validator~$v_i$. We observe two cases:

\begin{description}
    \item[Case 1: $c'_\mathsf{j}=c_\mathsf{f}+1$.] If $\chain'_\mathsf{j}$ conflicts with $\chain_\mathsf{f}$, then $A_i \preceq \chain'_\mathsf{j}$ implies $A_i \neq \chain_\mathsf{f}$, and thus the two votes are different. Conversely, if $\chain'_\mathsf{j}$ is a strict prefix of $\chain_\mathsf{f}$, then $A_i \preceq \chain'_\mathsf{j} {\prec} \chain_\mathsf{f}$, and thus $A_i \neq \chain_\mathsf{f}${, which implies that the two votes are different in this case as well}. {Hence, in} both cases, validator~$v_i$ {violates condition $\mathbf{E_1}$ and therefore} is slashable.

    \item[Case 2: $c'_\mathsf{j}>c_\mathsf{f}+1$.] We have to consider three cases:

\begin{description}
    \item[Case 2.1: $c_i > c_\mathsf{f}$.] 
    {By the definition of justification, if the checkpoint $(\chain'_\mathsf{j},c'_\mathsf{j})$ is justified, then the checkpoint $(A_i, c_i)$ must also be justified.} However, {given that $c_i < c'_\mathsf{j}$,} under the assumption of the minimality of $c'_\mathsf{j}$, we have $\chain_\mathsf{f} \preceq A_i$. Given that $A_i \preceq \chain'_\mathsf{j}$, this leads to the contradiction $\chain_\mathsf{f} \preceq \chain'_\mathsf{j}$.

    \item[Case 2.2: $c_i = c_\mathsf{f}$.]
    As argued in the proof of the case above, the checkpoint $(A_i, c_i)$ is justified.
    As proved at the beginning of this proof, either $A_i = \chain_\mathsf{f}$, or at least $\frac{n}{3}$ validators have violated condition $\mathbf{E_1}$.
    Consequently $\chain_\mathsf{f} = A_i \preceq \chain'_\mathsf{j}$ leading to a contradiction.

    \item[Case 2.3: $c_i < c_\mathsf{f}$.] 
    Considering $c_i < c_\mathsf{f} < c_\mathsf{f}+1 < c'_\mathsf{j}$, this situation constitutes a violation of $\mathbf{E_2}$ due to the existence of surrounding voting.
\end{description}
\end{description}
{Given that justifications require supermajority link, the intersection of the set of voters involved in the finalization of $\C_\mathsf{f}$ and justification of $\C'_\mathsf{j}$ includes at least $\frac{n}{3}$ validators which, by the reasoning outlined above, would violated either $\mathbf{E_1}$ or $\mathbf{E_2}$.}
This leads to contradicting our assumptions which concludes the proof.
\end{proof}

\begin{lemma}
\label{lem:accountable-safety}
Assume $f \in [0,n]$.
If two conflicting chains are finalized according to any two respective views, then at least $\frac{n}{3}$ validators can be detected to have violated either $\mathbf{E_1}$ or $\mathbf{E_2}$.
\end{lemma}

\begin{proof}
Let $\chain_1$ and $\chain_2$ be two conflicting chains according to view $\V_1$ and $\V_2$ respectively.
By the definition of finalized chains, this implies that there exists two checkpoints $\C_1 = (\chain'_1, c_1)$ and $\C_2 = (\chain'_2, c_2)$ that are finalized according to view $\V_1$ and $\V_2$ respectively, such that $\chain'_1 \succeq \chain_1$ conflicts with $\chain'_2 \succeq \chain_2$.
Assume without loss of generality that $c_2 \geq c_1$.
Given that finalization implies justification, we can apply \Cref{lem:accountable-jutification} to conclude that if $\chain'_2 \not\succeq \chain'_1$, then at least $\frac{n}{3}$ validators can be detected to have violated either $\mathbf{E_1}$ or $\mathbf{E_2}$.
\end{proof}

\section{Analysis of \protocol}
\label{appendix:analysis-3sf}

\subsection{Synchrony}\label{sec:analysis-tob-sync}

\def\FFGExec{\ensuremath{{e_{\textsc{fin}}}}}
\def\NoFFGExec{\ensuremath{{e_{\textsc{fin}}}}}

Throughout this part of the analysis, we assume $\GST=0$, that less than one-third of the entire validator set is ever controlled by the adversary (\ie, $f<\frac{n}{3}$), and that Constraint~\eqref{eq:pi-sleepiness} holds.

We start by formalizing the concept of validators sending the same messages except only for their finality component.
Moreover, we define the equivalence between an execution in \Cref{algo:prob-ga-fast} and one in \Cref{alg:3sf-tob-noga}.
These definitions will be leveraged upon in the subsequent Lemmas and Theorems.

\begin{definition}\label{def:dyn-equiv}
  We say that two messages are \emph{dynamically-equivalent} if and only if they differ in at most their finality-related components, \ie, they either are \textsc{propose} messages and differ at most in the greatest justified checkpoint component or they are \textsc{vote} messages and differ at most in the \textsc{fin-vote} component.
  
  We say that two sets of messages are dynamically-equivalent if and only if they have the same set of equivalence classes under type-equivalence, \ie,
  for any message $m$ in any of the two sets, there exists a message $m'$ in the other set such that $m$ and $m'$ are dynamically-equivalent.

  Given two executions $e$ and $e'$ and round $r$, we say that the two   executions are \emph{honest-output-dynamically-equivalent up to round $r$} if and only if for any round $r_j<r$ and validator $v_j$ honest in round $r_j$ in both executions, the set of messages sent by $v_j$ in round $r_j$ in execution $e$ is dynamically-equivalent to the set of messages sent by $v_j$ in round $r_j$ in execution $e'$.
\end{definition}

In the remainder of this section, we use the notation $\specifyExec{e}{\X}$, where $\X$ is any variable or definition, 
to explicitly indicate the value of $\X$ in execution $e$.
Due to the difference in the fast confirmation function employed by the two algorithms, we do need to explicitly specify the meaning of $\mfcproposeExec{e}{t}_i$ with $t$ being any slot.
If $e$ is an execution of \Cref{algo:prob-ga-fast}, then $\mfcproposeExec{e}{t}_i$ corresponds the definition provided in \Cref{appendix:analysis-prob-ga-fast}.
If $e$ is an execution of \Cref{alg:3sf-tob-noga}, then $\mfcproposeExec{e}{t}_i := \specifyExec{e}{\mfc}(\V^{\proposing{t}}_i,\V^{\proposing{t}}_i,\chain^C,t)$ with $(\chain^C,\cdot) = \texttt{fastconfirm}(\V^\proposing{t}_i,t-1)$.

{
\begin{definition}\label{def:exec-equiv}
    Let $\FFGExec$ by a compliant execution of \Cref{alg:3sf-tob-noga} and $\NoFFGExec$ be a compliant execution of \Cref{algo:prob-ga-fast}.
    We say that $\FFGExec$ and $\NoFFGExec$ are \emph{\mfc-equivalent} if and only if the following constraints hold:
  \begin{enumerate}

    \item\label[condition]{cond:0b-2-t} $\FFGExec$ and $\NoFFGExec$ are honest-output-dynamically-equivalent up to any round
    \item for any slot $t_i$,
    \begin{enumerate}[label*=\arabic*]
      \item \label[condition]{cond:3-2} $\mfcproposeNoFFG{t_i}_i = \mfcproposeFFG{t_i}_i$
      \item\label[condition]{cond:4-2} $\mfcvoteNoFFG{t_i}_i = \mfcvoteFFG{t_i}_i$
      \item\label[condition]{cond:5-2} for any slot $t_j \leq t_i$ and validator $v_j \in H_{\voting{(t_j)}}$, there exists a slot $t_k \leq t_j$ and a validator $v_k \in H_{\voting{(t_k)}}$ such that $\Chain^{\voting{t_k}}_k \succeq \chainava^{\voting{t_j}}_j$.
      \item\label[condition]{cond:6-2} for any slot $t_j \leq t_i$ and validator $v_j \in H_{\voting{(t_j)}}$, there exists a slot $t_k \leq t_j$ and a validator $v_k \in H_{\voting{(t_k)}}$ such that $\Chain^{\voting{t_k}}_k \succeq \chainava^{\fastconfirming{t_j}}_j$.
    \end{enumerate}
    where chain $\chainava_j$ is in \FFGExec{}, and chain $\Chain_i$ is in \NoFFGExec{}.
  \end{enumerate}
\end{definition}
}

{
\begin{lemma}\label{lem:equiv-ga2}
    Let $\FFGExec$ by any compliant execution of \Cref{alg:3sf-tob-noga}.  
    There exists a compliant execution $\NoFFGExec$ of \Cref{algo:prob-ga-fast} 
    such that $\FFGExec$ and $\NoFFGExec$ are \mfc-equivalent.
\end{lemma}
}
\def\ANoFFG{\ensuremath{A^{\NoFFGExec}}}
\def\AFFG{\ensuremath{A^{\FFGExec}}}
\begin{proof}
  Let $\FFGExec$ be any compliant execution of \Cref{alg:3sf-tob-noga} and let $\AFFG$ be the adversary in such an execution.
  Below, we specify the decisions made by an adversary $\ANoFFG$ in an execution $\NoFFGExec$ of \Cref{algo:prob-ga-fast}.
  Later we show that this leads to $\NoFFGExec$ satisfying all the conditions in the Lemma.
  \begin{enumerate}[label=(\roman*)]
    \item For any round $r$, the set of validators corrupted by $\ANoFFG$ in round $r$ corresponds to set of validators corrupted by $\AFFG$ in the same round $r$. 
    \item For any round $r$, the set of validators put to sleep by $\ANoFFG$ in round $r$ corresponds to set of validators put to sleep by $\AFFG$ in the same round $r$. 
    \item For any slot $t$, if an honest validator is the proposer for slot $t$ in $\FFGExec$, then it is also the proposer for slot $t$ in $\NoFFGExec$.
    \item\label{cond:adv-decision-4} For any round $r$, if $\FFGExec$ and $\NoFFGExec$ are honest-output-dynamically-equivalent up to round $r$, then $\ANoFFG$ schedules the  delivery of messages to the validators honest in round $r$ such that, for any validator $v_i$ honest in round $r_i$, the set of messages received by $v_i$ in execution $\FFGExec$ is dynamically-equivalent to the set of messages received by $v_i$ in execution $\NoFFGExec$.
  \end{enumerate}  

  Now, we prove by induction on $t_i$ that the execution $\NoFFGExec$ induced by $\ANoFFG$ satisfies all \Cref{def:exec-equiv}'s conditions.
  To do so, we add the following conditions to the inductive hypothesis.
  \begin{enumerate}[start=3]
    \item[]%
    \begin{enumerate}[start=5,label*=\arabic*]
      \item \label[condition]{cond:1-2} For any $\J$ such that $\mathsf{J}(\J,\VFFG^{\proposing{t_i}}_i)$, $\mfcproposeNoFFG{t_i}_i \succeq \J.\chain$.
      \item\label[condition]{cond:2-2} For any $\J$ such that $\mathsf{J}(\J,\VFFG^{\voting{t_i}}_i)$, $\mfcvoteNoFFG{t_i}_i \succeq \J.\chain$.
    \end{enumerate}
  \end{enumerate} 
  and rephrase \Cref{cond:0b-2-t} as follows
  \begin{enumerate}[start=3]
    \item[]%
    \begin{enumerate}[start=7,label*=\arabic*]
      \item\label[condition]{cond:0bt-2}  $\FFGExec$ and $\NoFFGExec$ are honest-output-dynamically-equivalent up to round $4\Delta (t_i+1)$.
    \end{enumerate}
  \end{enumerate}

  Note that \Cref{cond:0bt-2} holding for any $t_i$ implies \Cref{cond:0b-2-t}.

  To reduce repetitions in the proof, we treat the base case of the induction within the inductive step.
  Also, we let $\T^\voting{t_\ell}_j$ be the target checkpoint used by validator $v_j$ when sending an \textsc{fin-vote} in slot $t_\ell$.

  Then, take any $t_i \geq 0$ and assume that, if $t_i > 0$, then the Lemma and the additional conditions \Cref{cond:1-2,cond:2-2,cond:0bt-2} hold for slot $t_i-1$.
  We prove that they also hold for slot $t_i$.
  We start by proving \Cref{cond:1-2,cond:2-2} for slot $t_i$, 
  then we move to \Cref{cond:3-2,cond:4-2,cond:5-2,cond:6-2} in this order and conclude with proving \Cref{cond:0bt-2}.

 \begin{description}
  \item[\Cref{cond:1-2,cond:2-2}.]
  Let $\J$ be any checkpoint such that  $\mathsf{J}(\J,\VFFG^{\voting{t_i}}_i)$.
  Because the target checkpoint of an \textsc{fin-vote} sent by an honest validator $v_\ell$ in round $r_{\ell}$ corresponds to $(\makeFFG{\T}^{r_{\ell}}_\ell.\chain,\cdot)$ and we assume $f<\frac{n}{3}$,
  there exists a slot $t_k \in [0,t_i)$ and a validator $v_k \in H_{\voting{(t_k)}}$ such that $\T^{\voting{t_k}}_k \succeq \J.\chain$.
  By the inductive hypothesis, \Cref{cond:5-2} implies that there exists a slot $t_m \in [0,t_k]$ and validator $v_m \in H_{\voting{(t_m)}}$ such that $\Chain^{\voting{t_m}}_m \succeq \chainava^{\voting{t_k}}_k$.
  \sloppy{Also, due to \Cref{line:algtob-vote-chainava,line:algtob-rejustification,line:algtob-set-target-checkpoint}, celarly  $\chainava^{\voting{t_k}}_k \succeq \T^{\voting{t_k}}_k$.
  Given that $t_m \leq t_k < t_i$, from \Cref{lem:ga-confirmed-always-canonical}, we know that $\mfcproposeNoFFG{t_i} \succeq \Chain^{\voting{t_m}}_m \succeq \chainava^{\voting{t_k}}_k \succeq \T^{\voting{t_k}}_k \succeq \J.\chain$ and that $\mfcvoteNoFFG{t_i} \succeq \Chain^{\voting{t_m}}_m \succeq \chainava^{\voting{t_k}}_k \succeq \T^{\voting{t_k}}_k \succeq \J.\chain$.
  Given that the set of justified checkpoints in $\VFFG^{\voting{t_i}}_i$ is a superset of the set of justified checkpoints in $\VFFG^{\proposing{t_i}}_i$, both \Cref{cond:1-2} and \Cref{cond:2-2} are proven.}
  \item[\Cref{cond:3-2,cond:4-2}.]
  \sloppy{We know that $\FFGExec$ and $\NoFFGExec$ are honest-output-dynamically-equivalent up to round $\proposing{t_i}$.}
  If $t_i > 0$, then, by the inductive hypothesis, \Cref{cond:0bt-2} implies this, if $t=0$, then this is vacuously true.
  Hence, due to \cref{cond:adv-decision-4} of $\ANoFFG$'s set of decisions, $\VFFG^{\proposing{t_i}}_i$ and $\VNoFFG^{\proposing{t_i}}_i$ are dynamically-equivalent.
  From this follows that $\FFGExec$ and $\NoFFGExec$ are also honest-output-dynamically-equivalent up to round $\voting{t_i}$ and, therefore, $\VFFG^{\voting{t_i}}_i$ and $\VNoFFG^{\voting{t_i}}_i$ are dynamically-equivalent as well.
  This and \Cref{cond:1-2,cond:2-2} imply \Cref{cond:3-2,cond:4-2}.
  \item[\Cref{cond:5-2}.]
  \sloppy{By \Cref{line:algtob-vote-chainava}, 
  $\chainava^{\voting{t_i}}_i \in \{\chainava^{\fastconfirming{(t_i-1)}}_i,\GJfrozen[\voting{t_i}]_i.\chain,(\mfcvoteFFG{t_i}_i)^{\lceil \kappa}\}$.
  Let us consider each case.}
  \begin{description}
    \item[Case 1: $\chainava^{\voting{t_i}}_i = \chainava^{\fastconfirming{(t_i-1)}}_i$.] 
    This case implies $t_i > 0$.
    Hence, we can apply the inductive hypothesis.
    Specifically, \Cref{cond:6-2} for slot $t_i-1$ implies 
    \Cref{cond:5-2} for slot $t_i$.
    \item[Case 2: {$\chainava^{\voting{t_i}}_i = \GJfrozen[t_i]_i.\chain$}]
    Due to \Cref{line:algtob-merge-ch-frozen,line:algtob-prop-if,line:algtob-prop-merge-gj}, $\mathsf{J}(\GJfrozen[t_i]_i, \V^{\voting{t_i}})$. 
    Hence, by following the reasoning applied when discussing \Cref{cond:1-2,cond:2-2},
    there exists a slot $t_k \in [0,t_i)$ and validator $v_k \in H_{\voting{(t_k)}}$ such that  $\T^{\voting{t_k}}_k\succeq \chainava^{\voting{t_i}}_i$ from which we can conclude that $\chainava^{\voting{t_k}}_k\succeq \chainava^{\voting{t_i}}_i$.

    Therefore, we can apply the inductive hypothesis.
    Specifically, \Cref{cond:5-2} for slot $t_k$ implies 
    \Cref{cond:5-2} for slot $t_i$.
    \item[Case 3: {\normalfont $\chainava^{\voting{t_i}}_i = (\mfcvoteFFG{t_i}_i)^{\lceil \kappa}$}.] 
    \Cref{line:algga-no-fin-vote-chainava} of \Cref{algo:prob-ga-fast} implies that $\Chain^{\voting{t_i}}_i \succeq \mfcvoteNoFFG{t_i}_i$.
    Then, 
    \Cref{cond:4-2} implies that $\Chain^{\voting{t_i}}_i \succeq \chainava^{\voting{t_i}}_i$.
  \end{description}
  \item[\Cref{cond:6-2}.]
  By \Crefrange{line:algotb-at-confirm}{line:algtob-set-chaava-to-bcand}, $\chainava^{\fastconfirming{t_i}}_i \in \{\chainava^{\voting{t_i}}_i, \GJ(\VFFG^{\fastconfirming{t_i}}).\chain, \chain^C\}$ with $(\chain^C, Q) = \texttt{fastconfirm}(\VFFG^{\fastconfirming{t_i}}_i,t_i) \land Q\neq \emptyset$.
  Let us consider each case.
  \begin{description}
    \item[Case 1: $\chainava^{\fastconfirming{t_i}}_i = \chainava^{\voting{t_i}}_i$.] In this case, \Cref{cond:5-2} implies \Cref{cond:6-2}.
    \item[Case 2: \normalfont$\chainava^{\fastconfirming{t_i}}_i = \GJ(\VFFG^{\fastconfirming{t_i}}_i).\chain$.] 
    By following the reasoning applied when discussing \Cref{cond:1-2,cond:2-2}, this means that there exists a slot $t_k \in [0,t_i]$ and validator $v_k \in H_{\voting{(t_ k)}}$ such that $\chainava^{\voting{t_k}}_k \succeq \chainava^{\fastconfirming{t_i}}_i$.
    Then, \Cref{cond:5-2} for slot $t_k$ implies 
    \Cref{cond:6-2} for slot $t_i$.
    \item[Case 3: $\chainava^{\fastconfirming{t_i}}_i = \chain^C$ with $(\chain^C, Q) = \texttt{fastconfirm}(\VFFG^{\fastconfirming{t_i}}_i,t_i) \land Q\neq \emptyset$.]
    This implies that $\VFFG^{\fastconfirming{t_i}}_i$ includes a quorum of \textsc{vote} message for chain $\chain^C$.
    Given \Cref{cond:4-2},  $\VNoFFG^{\fastconfirming{t_i}}_i$ also includes a quorum of \textsc{vote} message for chain $\chain^C$.
    Hence, $\Chain^{\fastconfirming{t_i}}_i = \chainava^{\fastconfirming{t_i}}_i$. 
    \end{description}
  \item[\Cref{cond:0bt-2}.]
  As argued in the proof of \Cref{cond:3-2,cond:4-2}, we know that $\FFGExec$ and $\NoFFGExec$ are honest-output-dynamically-equivalent up to round $\proposing{t_i}$.
  This, \Cref{cond:3-2,cond:4-2} and \cref{cond:adv-decision-4} of $\ANoFFG$'s set of decisions clearly imply \Cref{cond:0bt-2} up to round $\proposing{(t_i+1)}$.\qed
 \end{description}
\end{proof}

In the following Lemmas and Theorems, unless specified, we refer to executions of \Cref{alg:3sf-tob-noga}.

\begin{lemma}[Analogous of \Cref{lem:keep-voting-tob-fast-conf}]
  \label{lem:keep-voting-tob-fast-conf-ffg}
  If, in slot $t$, all validators in $H_{\voting{t}}$ send \textsc{vote} messages for chains extending chain \(\chain\),
  then, for any validator $v_i\in H_{\voting{(t+1)}}$, $\mfcpropose{(t+1)}_i \succeq \chain$ and $\mfcvote{(t+1)}_i \succeq \chain$, which implies that, in slot $t+1$, all validators in $H_{\voting{(t+1)}}$ send \textsc{vote} messages for chains extending \(\chain\).
\end{lemma}
\begin{proof}
  Take any compliant execution $\FFGExec$ of \Cref{alg:3sf-tob-noga} where in slot $t$, all validators in $H_{\voting{t}}$ send \textsc{vote} messages for chains extending chain \(\chain\).
  Let $\NoFFGExec$ be an \mfc-equivalent execution of \Cref{algo:prob-ga-fast} which \Cref{lem:equiv-ga2} proves to exist.
  It is clear from the definition of $\mfc$-equivalent honest-output-dynamically-equivalent executions that in slot $t$ of execution $\NoFFGExec$ all validators in $H_{\voting{t}}$ send \textsc{vote} messages for chains extending chain \(\chain\) as well.
  Hence, by \Cref{lem:keep-voting-tob-fast-conf}, in execution $\NoFFGExec$, for any validator $v_i\in H_{\voting{(t+1)}}$, $\mfcpropose{(t+1)}_i \succeq \chain$ and $\mfcvote{(t+1)}_i \succeq \chain$.
  From \Cref{cond:3-2,cond:4-2} of \Cref{lem:equiv-ga2} then it follows that this holds for execution $\FFGExec$ as well, which implies that, in slot $t+1$, all validators in $H_{\voting{(t+1)}}$ send \textsc{vote} messages for chains extending \(\chain\).
\end{proof}

\begin{lemma}[Analogous of \Cref{lem:vote-proposal-fast-conf}]
  \label{lem:vote-proposal-fast-conf-ffg}
  Let $t$ be a slot with an honest proposer $v_p$ and assume that~$v_p$ sends a $[\textsc{propose}, \chain_p, \chain^C_p, Q^C_p, t, v_p]$ message. 
  Then, for any slot $t'\geq t$, all validators in $H_{\voting{(t')}}$ send a \textsc{vote} message for a chain extending $\chain_p$.
  Additionally, for any slot $t'' > t$ and any validator $v_i\in H_{\voting{(t'')}}$, $\mfcpropose{t''}_i\succeq \chain_p$ and $\mfcvote{t''}_i\succeq \chain_p$.
\end{lemma}
\begin{proof}
  Take any compliant execution $\FFGExec$ of \Cref{alg:3sf-tob-noga} where $v_p$, the proposer of slot $t$, is honest and sends a $[\textsc{propose}, \chain_p, \chain^C_p, Q^C_p t, v_p]$ message.
  Let $\NoFFGExec$ be an $\mfc$-equivalent execution of \Cref{algo:prob-ga-fast} which \Cref{lem:equiv-ga2} proves to exist.
  It is clear from the definition of $\mfc$-equivalent honest-output-dynamically-equivalent executions that $v_p$ sends the same \textsc{propose} message in slot $t$ of execution $\NoFFGExec$.
  Hence, by \Cref{lem:vote-proposal-fast-conf}, in execution $\NoFFGExec$, for any slot $t'\geq t$, all validators in $H_{\voting{(t')}}$ send a \textsc{vote} message for a chain extending $\chain_p$, and, for any slot $t'' > t$ and any validator $v_i\in H_{\voting{(t'')}}$, $\mfcpropose{t''}_i\succeq \chain_p$ and $\mfcvote{t''}_i\succeq \chain_p$.
  From \Cref{cond:0b-2-t,cond:3-2,cond:4-2} of \Cref{lem:equiv-ga2} then it follows that this holds for execution $\FFGExec$ as well.
\end{proof}

\begin{lemma}[Analogous of \Cref{lem:ga-confirmed-always-canonical}]\label{lem:ga-confirmed-always-canonical-ffg}
  Let $r_i$ be any round and $r_j$ be any round such that $r_j\ge r_i$ and $r_j \in \{\proposing{r_j},\voting{r_j}\}$. Then, for any validator~$v_i$ honest in round $r_i$ and any validator $v_j \in H_{r_j}$, 
  $\chainava^{r_i}_i \preceq \mfc^{r_j}_j$.
\end{lemma}
\begin{proof}
  Take any compliant execution $\FFGExec$ of \Cref{alg:3sf-tob-noga}.
  Let $\NoFFGExec$ be an equivalent execution of \Cref{algo:prob-ga-fast} which \Cref{lem:equiv-ga2} proves to exist.
  Let $r_i$ be any round and $r_j$ be any round such that $r_j\ge r_i$ and $r_j \in \{\proposing{r_j},\voting{r_j}\}$. 
  We know that $\chainava^{r_i}_i$ is first set by $v_i$ in round $r_i' \leq r_i$ such that $r_i' \in \{\voting{\slot(r_i')},\fastconfirming{\slot(r_i')}\}$.
  From \Cref{cond:5-2,cond:6-2} of \Cref{lem:equiv-ga2}, we know that there exits a slot $t_m \leq r_k$ and a validator $v_m \in H_{\voting{(t_m)}}$ such that $\Chain^{\voting{t_m}}_m \succeq \chainava^{r_i'}_i = \chainava^{r_i}_i$.
  \Cref{lem:ga-confirmed-always-canonical} implies that $\mfcNoFFG^{r_j}_j \succeq \Chain^{\voting{t_m}}_m$.
  From \Cref{lem:equiv-ga2}, we know that $\mfcFFG^{r_j}_j= \mfcNoFFG^{r_j}_j$.
  Hence,  $\mfcFFG^{r_j}_j= \mfcNoFFG^{r_j}_j \succeq \Chain^{\voting{t_m}}_m\succeq \chainava^{r_i'}_i = \chainava^{r_i}_i$.
\end{proof}

\begin{theorem}[Reorg Resilience - Analogous of \Cref{thm:reorg-res-prop-tob}]\label{thm:reorg-res-prop-tob-ffg}
  \Cref{alg:3sf-tob-noga} is reorg-resilient.
\end{theorem}
\begin{proof}
  The proof follows the one for \Cref{thm:reorg-res-prop-tob} with the following changes only.
  \begin{enumerate}
    \item Replace $\Chain$ with $\chainava$,
    \item Replace \Cref{lem:vote-proposal-fast-conf} with \Cref{lem:vote-proposal-fast-conf-ffg}, and
    \item Replace \Cref{lem:ga-confirmed-always-canonical} with \Cref{lem:ga-confirmed-always-canonical-ffg}.\qed
  \end{enumerate}
\end{proof}

\begin{theorem}[dynamic-availability - Analogous of \Cref{thm:dyn-avail-fast-conf-tob}]
  \label{thm:dyn-avail-fast-conf-tob-ffg}
\Cref{alg:3sf-tob-noga} is dynamically available \wop.
\end{theorem}

\begin{proof}
  {Note that \Cref{lem:ga-mfc-proposer-shorter-than-t} holds for \Cref{alg:3sf-tob-noga} as well.}
  Then, the proof follows the one for \Cref{thm:dyn-avail-fast-conf-tob} with the following changes only.
  \begin{enumerate}
    \item Replace {$\Chain$ with $\chainava$},
    \item Replace the reference to \Cref{line:algga-no-fin-vote-chainava} of \Cref{algo:prob-ga-fast} with \Cref{line:algtob-vote-chainava} of \Cref{alg:3sf-tob-noga},
    \item Replace \Cref{lem:vote-proposal-fast-conf} with \Cref{lem:vote-proposal-fast-conf-ffg},
    \item Replace \Cref{lem:ga-confirmed-always-canonical} by \Cref{lem:ga-confirmed-always-canonical-ffg}, and
    \item In the proof of liveness, consider the following additional case.
    Assume that $r_i \geq \voting{(t+2\kappa)}$ is a fast confirmation round and that $\chainava^{r_i}_i$ is set to $\GJ(\V^{r_i}_i).\chain$ via \Cref{line:algtob-set-chaava-to-bcand}.
    Because the chain of the target checkpoint of an \textsc{fin-vote} sent by an honest validator $v_\ell$ in round $r_{\ell}$ corresponds to $(\makeFFG{\T}^{r_{\ell}}_\ell.\chain,\cdot)$, with $\makeFFG{\chainava}^{r_{\ell}}_\ell.\chain\succeq\makeFFG{\T}^{r_{\ell}}_\ell.\chain$, and we assume $f<\frac{n}{3}$, there exists a slot $t_k \in [0,\slot(r_i)]$ and a validator $v_k \in H_{\voting{(t_k)}}$ such that $\chainava^{\voting{t_k}}_k \succeq \GJ(\V^{r_i}_i).\chain = \chainava^{r_i}_i$.
    We have already established that $\chainava^{\voting{\slot(r_i)}}_i \succeq \chain_p$.
    Then, \Cref{lem:ga-confirmed-always-canonical-ffg} implies that $\mfcvote{\slot(r_i)}_i \succeq \GJ(\V^{r_i}_i).\chain =  \chainava^{r_i}_i $ and $\mfcvote{\slot(r_i)}_i \succeq \chainava^{\voting{\slot(r_i)}}_i \succeq \chain_p$.

    This implies that $\chainava^{\voting{\slot(r_i)}}_i$ and $\GJ(\V^{r_i}_i).\chain$ do not conflict.
    Hence, given that we assume that \Cref{line:algtob-set-chaava-to-bcand} is executed, $\chainava^{\voting{\slot(r_i)}}_i \prec \GJ(\V^{r_i}_i).\chain$.
    Because, $\chainava^{\voting{\slot(r_i)}}_i \succeq \chain_p$, we can conclude that 
    $\chainava^{r_i}_i = \GJ(\V^{r_i}_i).\chain \succ \chainava^{\voting{\slot(r_i)}}_i\succeq \chain_p$.\qed
  \end{enumerate}
\end{proof}

{
\begin{lemma}[Analogous of \Cref{thm:fast-liveness-tob}]
\label{thm:fast-liveness-modified-tob}
Take a slot $t$ in which \(|H_{\voting{t}}| \geq \frac{2}{3}n\).
If in slot $t$ an honest validator sends a \textsc{propose} message for chain $\chain_p$, then, for any validator $ v_i \in H_{\fastconfirming{t}}$, $\chainava^{\fastconfirming{t}}_i \succeq \chain_p$.
\end{lemma}

\begin{proof}
    The proof follows the one for \Cref{thm:fast-liveness-tob} {with the only change being replacing $\Chain_i$ with $\chainava_i$}.
\end{proof}
}

\subsection{Partial synchrony}
\label{sec:analysis-tob-psync}
In this section, we show that \Cref{alg:3sf-tob-noga} ensures that the chain $\chainfin$ is always Accountably Safe and is live after time $\max(\GST, \GAT) + {\Delta}$, meaning that \Cref{alg:3sf-tob-noga} is a secure ebb-and-flow protocol {under the assumption that $f<\frac{n}{3}$}.

We start by showing that honest validators following \Cref{alg:3sf-tob-noga} are never slashed.
 
\begin{lemma} \label{lem:never-slashed-3sf-tob-noga}
Honest validators following \Cref{alg:3sf-tob-noga} are never slashed.
\end{lemma}

\begin{proof}
First of all, observe that

  \begin{enumerate}
    \item By  \Cref{line:algtob-vote}, an honest active validator sends only one \textsc{fin-vote} in any slot.
    \item As a consequence of \Cref{line:algtob-set-target-checkpoint}, if $\T$ is the target checkpoint in the \textsc{fin-vote} sent in slot $t$, then $\T.c = t$.
    \item Take any two slots $t$ and $t'$ with $t< t'$ such that $v_i \in H_\voting{t}$ and $v_i \in H_\voting{t'}$.
    Let $\calS^\ell$ be the source checkpoint of the \textsc{fin-vote} that validator $v_i$ sent in any slot $\ell$.
    Then $\calS^t.c \leq \calS^{t'}.c$.
    Specifically, take any two slots $k$ and $k'$ such that $t \leq k < k' \leq t'$, $v_i \in H_{\voting{k}}$ and $k'$ is the first slot after $k$ such that $v_i \in H_{\voting{k'}}$.
    By \Cref{line:algtob-vote}, we know that $\calS^k = \GJfrozen[\voting{k}]_i$ and by \Cref{line:algtob-prop-if,line:algtob-merge-ch-frozen}, we have that $\mathsf{J}( \GJfrozen[\voting{k}]_i,\V^\voting{k}_i)$.
    Then, note that $k \leq k'-1$ and that because of the joining protocol, we know that $v_i$ is awake at round $\merging{(k'-1)}$.
    Hence, we have that  $\GJfrozen[\voting{k}]_i.c \leq \GJ(\V^\merging{(k'-1)}_i).c = \GJfrozen[\merging{(k'-1)}]_i.c$.
    Finally, by \Cref{line:algtob-prop-if,line:algtob-prop-merge-gj}, $\GJfrozen[\merging{(k'-1)}]_i.c \leq \GJfrozen[\voting{k'}]_i.c = \calS^{k'}.c$ and $\calS^{k}.c \leq \calS^{k'}.c$.
  \end{enumerate}

  Now, take any round $r$ and any validator $v_i$ honest in round $r$.
  Take also any two different \textsc{fin-vote}s $\C_1 \rightarrow \C_2$ and $\C_3 \rightarrow \C_4$ that $v_i$ has sent by round $r$.
  If no such \textsc{fin-vote}s exists, then clearly $v_i$ cannot be slashed.
  By 1 and 2 above we have that $\C_2.c \neq \C_4.c$.
  Hence, $\textbf{E}_\textbf{1}$ is not violated.
  Then, without loss of generality, assume $\C_2.c < \C_4.c$.
  This means that \textsc{fin-vote} $\C_1 \rightarrow \C_2$ has been sent before  \textsc{fin-vote} $\C_3 \rightarrow \C_4$.
  Then, by 3 above, we know that $\C_1.c  \leq  \C_3.c$.
  Given that $\C_2.c < \C_4.c$, $\textbf{E}_\textbf{2}$ cannot be violated.
  Therefore, until $v_i$ is honest, it cannot be slashed.
  \qed
\end{proof}

Let $\Vglobal^r$ denote the global view at time $r$, that is, the set of all messages sent up to and including time $r$. The chain $\chainfin^r_i$ output by an honest validator $v_i$ in round $r$ of \Cref{alg:3sf-tob-noga} is any chain that is finalized according to the view $\Vglobal^r$. This directly follows from \Cref{line:algotb-set-chfin-init,line:algotb-set-chfin-vote,line:algotb-set-chfin-fast} and the definition of $\GF$.

\begin{theorem}[Accountable Safety]
  \label{thm:accountable-safety}
  Let $f \in [0,n]$. The finalized chain $\chainfin$ is {$\frac{n}{3}$-accountable}.
\end{theorem}

\begin{proof}
Follows from \Cref{lem:accountable-safety,lem:never-slashed-3sf-tob-noga}.
\end{proof}

We now show that the finalized chain of an honest validator grows monotonically. Although this property is not explicitly defined in \Cref{sec:security}, it is frequently desired by Ethereum protocol implementers~\cite{ethereum-properties}.

\begin{lemma}\label{lem:chfin-ga-always-grows}
  For any two round $r' \geq r$ and validator $v_i$ honest in round $r'$, $\chainfin^{r'}_i \succeq \chainfin^r_i$.
\end{lemma}
\begin{proof}
  Let us prove this Lemma by contradiction.
  Let $r$ be the smallest round such that there exist round $r' \geq r$ and a validator $v_i$ honest in round $r'$ such that $\chainfin^{r'}_i \nsucceq \chainfin^r_i$.
  {Clearly $r' > r$.}
  Assume $r'$ to be the smallest such round.
  First, we want to show that there exists a checkpoint $\J$ such that $\chainava^{r'}_i \succeq \J.\chain \land \mathsf{J}(\J, \V^{r'}_i) \land$ {$\J \geq \GF(\V^r_i)$}.
  Note that both $r$ are $r'$ are either a vote or a fast confirmation round.
  This is due to the minimality of $r$ and $r'$ and the fact that $\chainfin_i$ is only set in these types of rounds.
  Then, let us proceed by cases.

  \begin{description}
    \item[Case 1: $r'$ is a vote round.]
    In this case $\J = \GJfrozen[r']_i$.
    {Note that due to \Cref{line:algtob-merge-ch-frozen,line:algtob-merge-gj-frozen,line:algtob-prop-check-chfrozen-gj,line:algtob-prop-set-ch-frozen-to-gj}, $\chainfrozen[r']_i \succeq \GJfrozen[r']_i.\chain$.
    Then, \Cref{line:algtob-vote-chainava} implies that $\chainava^{r'}_i \succeq \GJfrozen[r']_i.\chain$.
    Also, \Cref{line:algtob-merge-gj-frozen,line:algtob-prop-if,line:algtob-prop-merge-gj} imply that $\mathsf{J}(\GJfrozen[r']_i,\V^{r'}_i)$.
    Finally, \Cref{line:algtob-prop-if,line:algtob-prop-merge-gj,line:algtob-merge-gj-frozen}, and $r' > r$ imply that 
    {$\GJfrozen[r']_i \geq \GJ(\V^{r}_i)$.}}
    \item[Case 2: $r'$ is a fast confirmation round.]
    In this case, $\J = \GJ(\V^{r'}_i)$.
    Clearly, $\mathsf{J}(\GJ(\V^{r'}_i),\V^{r'}_i)$.
    {\Crefrange{line:algtob-set-fastcand-fconf}{line:algtob-set-chaava-to-bcand} ensure that $\chainava^{r'}_i \succeq \GJ(\V^{r'}_i).\chain$.}
    Then, from \Cref{alg:justification-finalization} and the fact that $\V^r_i\subseteq \V^{r'}_i$, we can conclude that {$\GJ(\V^{r'}_i) \geq \GF(\V^{r'}_i) \geq \GF(\V^{r}_i).c$}.
  \end{description}  
  From \Cref{alg:justification-finalization} and $\V^{r}_i \subseteq \V^{r'}_i$, we also have that {$\GF(\V^{r'}_i) \geq \GF(\V^r_i)$}.
  Then, \Cref{lem:accountable-jutification} and $f<\frac{n}{3}$ imply $\chainava^{r'}_i \succeq \J.\chain \succeq \GF(\V^r_i).\chain \succeq \chainfin^r_i$ and
  $\GF(\V^{r'}_i).\chain\succeq \GF(\V^r_i).\chain$, which further imply
  $\chainfin^{r'}_i \succeq \GF(\V^r_i).\chain$.
  Hence, $\chainfin^{r'}_i  \succeq \GF(\V^r_i).\chain \succeq \chainfin^r_i$ reaching a contradiction.
\end{proof}

Now let us move to the liveness property.

\begin{lemma}
  \label{lem:vote-proposal-fast-conf-ffg-only-base}
  Let $t$ be a slot with an honest proposer $v_p$  such that $\proposing{t} \geq \GST + \Delta$ and assume that~$v_p$ sends a $[\textsc{propose}, \chain_p, \chain^C_p, Q^C_p, \GJ_p, t, v_p]$ message.
  Then, for any validator $v_i \in H_{\voting{t}}$, $v_i$ sends a \textsc{vote} message for chain $\chain_p$ and $\GJfrozen[\voting{t}]_i = \GJ_p$.
\end{lemma}  
\begin{proof}
  \sloppy{Suppose that in slot \(t\), an honest proposer sends a $[\textsc{propose}, \chain_p, \chain^C_p, Q^C_p,\GJ_p, t, v_p]$ message.}
  Consider an honest validator \(v_i\) in \(H_{\voting{t}}\).
  Note that due to the synchrony assumption, $\V^{\merging{(t-1)}}_i \subseteq \V^{\proposing{t}}_p$.
  Hence, {$\GJ_p \geq \GJfrozen[\merging{(t-1)}]_i$}.
  Now, we want to show that  $\chain^{\text{frozen},\voting{t}}_i = \chain^C_p$.
  Let us consider two cases.
  \begin{description}
    \item[Case 1: {$\chainfrozen[\merging{(t-1)}]_i\succ \GJfrozen[\merging{(t-1)}]_i.\chain$}.]
    \sloppy{This implies that at time $\merging{(t-1)}$, $v_i$ has received a  quorum of \textsc{vote} messages for $\chainfrozen[\merging{(t-1)}]_i$.
    Let us consider three sub cases.}
    \begin{description}
      \item[Case 1.1: {$\chainfrozen[\merging{(t-1)}]_i = \GJ_p.\chain$}.] Given that $\chain^C_p \succeq \GJ_p.\chain$, \Cref{line:algtob-if-setchainfrozen-to-chainc,algo:3sf-noga-setchainfrozen-to-chainc} set  $\chain^{\text{frozen},\voting{t}}_i = \chain^C_p$.
      \item[Case 1.2: {$\chainfrozen[\merging{(t-1)}]_i \succ \GJ_p.\chain$}.] 
      This implies that $\V^{\proposing{t}}_p$ include a quorum of \textsc{vote} messages for $\chain^C_p$.
      Due to the synchrony assumption, the quorum of \textsc{vote} messages for $\chainfrozen[\merging{(t-1)}]_i$ are in the view of validator $v_p$ at time $\proposing{t}$.
      Given that $f<\frac{n}{3}$, this case implies that $\chainfrozen[\merging{(t-1)}]_i = \chain^C_p$, and therefore
      $\chainfrozen[\voting{t}]_i = \chain^C_p$.
      \item[Case 1.3: {$\chainfrozen[\merging{(t-1)}]_i \nsucceq \GJ_p.\chain$}.] In this case, \Crefrange{line:algtob-prop-check-chfrozen-gj}{algo:3sf-noga-setchainfrozen-to-chainc} set $\chain^{\text{frozen},\voting{t}}_i = \chain^C_p$.
    \end{description}
    \item[Case 2: {$\chainfrozen[\merging{(t-1)}]_i = \GJfrozen[\merging{(t-1)}]_i.\chain$}.] Given that $\chain^C_p \succeq \GJ_p$ and {$\GJ_p \geq \GJfrozen[\merging{(t-1)}]_i$}, \Crefrange{line:algtob-prop-if}{algo:3sf-noga-setchainfrozen-to-chainc} imply that $\chainfrozen[\voting{t}]_i = \chain^C_p$.
  \end{description}
  
  \sloppy{Now, first suppose that at round $\voting{t}$ there exists a chain $\chain \succeq \chain^{\text{frozen}}_i = \chain^C_p$ such that \(\left|\left(\Vfrozen\right)^{\chain,t} \cap \V_i^{\chain,t} \right| > \frac{\left|\mathsf{S}(\V_i,t)\right|}{2}\).}
  By the Graded Delivery property~\cite{streamliningSBFT}, this implies that at time $\proposing{t}$, $\left|\V_p^{\chain,t}\right|>\frac{\left|\mathsf{S}(\V_p,t)\right|}{2}$ meaning that $\chain_p\succeq \chain$ and hence, in slot $t$, $v_i$ sends a \textsc{vote} message for $\chain_p$.
  
  If no such a chain exists, $v_i$ still sends a \textsc{vote} for $\chain_p$ as $\chain_p \succeq \chain^C_p = \chain^\text{frozen}_i$.

  \medskip

  Also, given that {$\GJ_p \geq \GJfrozen[\merging{(t-1)}]_i$}, due to \Cref{line:algtob-prop-if,line:algtob-prop-merge-gj}, $\GJfrozen[\voting{t}] = \GJ_p$.
\end{proof}

\begin{lemma}
\label{lem:rejustification}
Assume that $f< \frac{n}{3}$.
Let $v_p$ be any validator honest by the time it \textsc{propose}s chain $\chain_p$ in a slot $t$ such that $\proposing{t} \geq \max(\GST, \GAT) + 4\Delta$ and assume that $\GJ_p^\proposing{t}.c < t -1$.
Then, chain $\chain_p$ is justified at slot~$t+1$, and finalized at slot~$t+2$.
\end{lemma}

\begin{proof}
Let $\chain_p$ be the chain proposed by the honest proposer at slot~$t$, with $\proposing{t} \geq \max(\GST, \GAT) + 4\Delta$, extending both the output of the fork-choice function $\mfc$ and the greatest justified checkpoint $\GJ_p$ with $\GJ_p^\text{frozen}.c < t - 1$.

By Line~\ref{line:algtob-rejustification} and \Cref{lem:vote-proposal-fast-conf-ffg-only-base}, each honest validator sets the target checkpoint to the same checkpoint $\T := (\GJ_p^\proposing{t}.\chain, t).$

Then, again due to \Cref{lem:vote-proposal-fast-conf-ffg-only-base}, in slot~$t$, all honest validators $v_i$ send \textsc{vote} messages of the form $[\textsc{vote},\, \chain_p,\, \GJ_p^\proposing{t} \to \T,\, t,\, v_i]$, which justify $\T$ (since $f < n/3$ ensures sufficient honest quorum).

Consequently, in slot~$t+1$, $\T$ is the greatest justified checkpoint in the view of all honest validators and is used as the source of the \textsc{fin-vote}s. \Cref{lem:vote-proposal-fast-conf-ffg-only-base} also implies that $\chainava_i^\voting{t+1} = \chain_p$.

Therefore, due to \Cref{line:algtob-set-target-checkpoint-if,line:algtob-set-target-checkpoint}, in slot $t+1$, any honest validator sends \textsc{fin-vote}s with source $\T$ and target $(\chain_p, t+1)$, thereby justifying $(\chain_p, t+1)$. In slot~$t+2$, all honest validators vote with source $(\chain_p, t+1)$, and by round $4\Delta(t+2) + 2\Delta$, these votes are visible to all honest validators, finalizing $\chain_p$ at slot~$t+2$.
\end{proof}

\begin{theorem}[Liveness]
\label{thm:liveness-chfin}
Assume $f < \frac{n}{3}$, and let $t$ and $t+1$ be slots after $\max(\GST, \GAT) + 4\Delta$. Then, the chain $\chain_p$ proposed by the proposer of slot $t$ is finalized at slot~$t+2$ or slot~$t+3$.

In particular, for any validator $v_i \in H_{r}$, where $r \ge 4\Delta(t+3) + 2\Delta$, it holds that $\chain_p \preceq \chainfin^{\,r}_i \quad \text{and} \quad \txpool^{\proposing{t}} \subseteq \chainfin^{\,r}_i.$
\end{theorem}

\begin{proof}
Let $v_{p_1}$ and $v_{p_2}$ be the honest proposers in slot $t$ and $t+1$, respectively. 
Let us first assume that $\GJ_p^\proposing{t}.c = t - 1$. 
We distinguish two cases:

\begin{enumerate}
    \item \textbf{Case 1: The slot of the greatest justified checkpoint in the view of proposer $v_{p_2}$ at propose time in slot $t+1$ is $t$, i.e., $\GJ_{p_2}^\proposing{t+1} = \C = (\chain',t)$ for some chain $\chain'$.} By \Cref{lem:vote-proposal-fast-conf-ffg-only-base}, and \Cref{line:algtob-set-target-checkpoint-if,line:algtob-set-target-checkpoint}, all honest validators $v_i$ in slot~$t+1$ have $\GJfrozen[\voting{t+1}]_i = \GJ_{p_2}^\proposing{t+1} = \C$, and send \textsc{fin-vote}s with source $\C$ and target $(\chain_p, t+1)$. Since $f < \frac{n}{3}$, these votes justify $(\chain_p, t+1)$ in slot~$t+1$. In slot~$t+2$, honest validators vote with source $(\chain_p, t+1)$. By round $4\Delta (t+2) + 2\Delta$, these votes are included in every honest validator's view, finalizing $\chain_p$ at slot~$t+2$. In particular, $\chain_p \preceq \GF(\V^{4\Delta(t+2)+2\Delta}).\chain$ and due to Line~\ref{line:algotb-set-chfin-fast}, $\chain_p \preceq \chainfin^{4\Delta(t+2)+2\Delta}_i$. Due to \Cref{lem:ga-mfc-proposer-shorter-than-t}, it follows that $\txpool^\proposing{t} \subseteq \chainfin^{4\Delta (t+2) + 2 \Delta}_i$. As a consequence of \Cref{lem:chfin-ga-always-grows}, $\chainfin^{4\Delta(t+2)+2\Delta}_i \preceq \chainfin^{4\Delta(t+3)+2\Delta}_i \preceq \chainfin^{r}_i$, for every $r\ge 4\Delta(t+3)+2\Delta$, and the result follows.

    \item \textbf{Case 2: No justification for slot $t$ is observed by $v_{p_2}$ at propose time in slot $t+1$.} In this case, we apply \Cref{lem:rejustification} to slot~$t+1$, which implies that the chain \textsc{propose}d by $v_{p_2}$ at slot~$t+1$ is justified in slot~$t+2$ and finalized in slot~$t+3$. Since this chain extends $\chain_p$ proposed by $v_{p_1}$ (due to \Cref{lem:vote-proposal-fast-conf-ffg}), $\chain_p$ is finalized as a prefix of the finalized chain in slot $t+3$. In particular, $\chain_p \preceq \GF(\V^{4\Delta(t+3)+2\Delta}).\chain$ and due to Line~\ref{line:algotb-set-chfin-fast}, $\chain_p \preceq  \chainfin^{4\Delta(t+3)+2\Delta}_i$. Due to \Cref{lem:ga-mfc-proposer-shorter-than-t}, it follows that $\txpool^\proposing{t} \subseteq \chainfin^{4\Delta (t+3) + 2 \Delta}_i$
\end{enumerate}
Observe that if $\GJ_p^\proposing{t}.c<  t - 1$, then we can immediately apply Lemma~\ref{lem:rejustification} finalizing the chain \textsc{propose}d by $v^t_p$ in slot~$t+2$. 
\end{proof}

\begin{theorem}\label{thm:ga-ebb-and-flow}
  \Cref{alg:3sf-tob-noga} is a secure ebb-and-flow protocol \wop.
\end{theorem}
\begin{proof}
  Safety of $\chainava$ and $\chainfin$, and  Liveness of $\chainava$ follow from this, \Cref{thm:accountable-safety,thm:dyn-avail-fast-conf-tob-ffg}, and \Cref{lem:never-slashed-3sf-tob-noga}.

  Then, we need to show that $\chainfin$ is live after $\max(\GST,\GAT) + O(\Delta)$ with confirmation time $O(\kappa)$.
  Let $t$ be any slot such that $\proposing{t} \geq \max(\GST,\GAT) + \Delta$.
  There is a high probability of finding a slot \( t_p  \in [t, t + \kappa)\) such that both $t_p$ and $t_p+1$ have honest proposers (\Cref{lem:2honestproposers}).
  By \Cref{thm:liveness-chfin}, $\txpool^\proposing{t}  \subseteq \txpool^\proposing{t_p} \subseteq \chainfin^\fastconfirming{(t_p+3)}_i$ for any validator honest in $\fastconfirming{(t_p+3)}$ and, clearly, $\fastconfirming{(t_p+3)}-\proposing{t} \in O(\kappa)$.

  {
    Finally, we need to show that, for any round $r$ and validator $v_i \in H_r$, $\chainfin^r_i \preceq \chainava^r_i$.
    \Cref{line:algotb-set-chfin-vote} clearly ensures that this is the case when $r$ is a vote round.
    Note that the only other times at which either $\chainava$ or $\chainfin$ are potentially modified are fast confirmation rounds.
    Then, let $r$ be any fast confirmation round.
    \Crefrange{line:algtob-set-fastcand-fconf}{line:algtob-set-chaava-to-bcand} ensure that $\GJ(\V^r_i).\chain \preceq \chainava^r_i$.
    By \Cref{alg:justification-finalization}, 
   {$\GJ(\V^r_i)\geq\GF(\V^r_i)$} 
    which, due to \Cref{lem:accountable-jutification} and the assumption that $f < \frac{n}{3}$, implies $\chainfin  \preceq \GJ(\V^r_i).\chain \preceq \chainava^r_i$.
  }
\end{proof}

\end{document}